\documentclass[11pt]{article}%
\usepackage{amsfonts}
\usepackage{amsmath}
\usepackage{fullpage}
\usepackage{hyperref}
\usepackage{amssymb}
\usepackage{graphicx}%
\providecommand{\U}[1]{\protect\rule{.1in}{.1in}}
\newtheorem{theorem}{Theorem}

\newtheorem{corollary}[theorem]{Corollary}

\newtheorem{definition}[theorem]{Definition}

\newtheorem{lemma}[theorem]{Lemma}

\newtheorem{proposition}[theorem]{Proposition}
\newtheorem{remark}[theorem]{Remark}

\newenvironment{proof}[1][Proof]{\noindent\textbf{#1.} }{\ \rule{0.5em}{0.5em}}
\numberwithin{equation}{section}
\begin{document}

\title{\textbf{Swiveled R\'{e}nyi entropies}}
\author{Fr\'ed\'eric Dupuis\thanks{Faculty of Informatics, Masaryk University, Brno,
Czech Republic}
\and Mark M. Wilde\thanks{Hearne Institute for Theoretical Physics, Department of
Physics and Astronomy, Center for Computation and Technology, Louisiana State
University, Baton Rouge, Louisiana 70803, USA}}
\maketitle

\begin{abstract}
This paper introduces \textquotedblleft swiveled R\'{e}nyi
entropies\textquotedblright\ as an alternative to the R\'{e}nyi entropic
quantities put forward in [Berta \textit{et al}., \textit{Phys.~Rev.~A}
\textbf{91}, 022333 (2015)]. What distinguishes the swiveled R\'{e}nyi
entropies from the prior proposal of Berta \textit{et al}.~is that there is an
extra degree of freedom:\ an optimization over unitary rotations with respect
to particular fixed bases (swivels). A consequence of this extra degree of
freedom is that the swiveled R\'{e}nyi entropies are ordered, which is an
important property of the R\'{e}nyi family of entropies. The swiveled
R\'{e}nyi entropies are however generally discontinuous at $\alpha=1$ and do
not converge to the von Neumann entropy-based measures in the limit as
$\alpha\rightarrow1$, instead bounding them from above and below. Particular
variants reduce to known R\'{e}nyi entropies, such as the R\'{e}nyi relative
entropy or the sandwiched R\'{e}nyi relative entropy, but also lead to ordered
R\'{e}nyi conditional mutual informations and ordered R\'{e}nyi
generalizations of a relative entropy difference. Refinements of entropy
inequalities such as monotonicity of quantum relative entropy and strong
subadditivity follow as a consequence of the aforementioned properties of the
swiveled R\'{e}nyi entropies. Due to the lack of convergence at $\alpha=1$, it
is unclear whether the swiveled R\'{e}nyi entropies would be useful in
one-shot information theory, so that the present contribution represents
partial progress toward this goal.

\end{abstract}

\section{Introduction}

In 1961, Alfred R\'{e}nyi defined a parametrized family of entropies now
bearing his name, by relaxing one of the axioms that singles out the Shannon
entropy \cite{R61}. This led to both the $\alpha$-R\'{e}nyi entropy and the
$\alpha$-R\'{e}nyi divergence, defined respectively for a parameter $\alpha
\in\left(  0,1\right)  \cup\left(  1,\infty\right)  $ and probability
distributions $p$ and $q$ as%
\begin{align}
H_{\alpha}(p)  &  \equiv\frac{1}{1-\alpha}\log\sum_{x}\left[  p( x) \right]
^{\alpha},\label{eq:Renyi-ent}\\
D_{\alpha}(p\Vert q)  &  \equiv\frac{1}{\alpha-1}\log\sum_{x}\left[  p( x)
\right]  ^{\alpha}\left[  q( x) \right]  ^{1-\alpha}, \label{eq:Renyi-rel-ent}%
\end{align}
where $\log$ denotes the natural logarithm here and throughout the paper. The
Shannon entropy and relative entropy are recovered in the limit as
$\alpha\rightarrow1$:%
\begin{align}
\lim_{\alpha\rightarrow1}H_{\alpha}(p)  &  =H(p)\equiv-\sum_{x}p( x) \log p(
x) ,\\
\lim_{\alpha\rightarrow1}D_{\alpha}(p\Vert q)  &  =D(p\Vert q)\equiv\sum_{x}p(
x) \log\frac{p( x) }{q( x) }.
\end{align}
What began largely as a theoretical exploration ended up having many practical
ramifications, especially in the contexts of information theory and
statistics. For example, it is now well known that the R\'{e}nyi entropies
play a fundamental role in obtaining a sharpened understanding of the
trade-off between communication rate, error probability, and number of
resources in communication protocols, such as data compression and channel
coding \cite{C95,vEH14}. \textquotedblleft Smoothing\textquotedblright\ the
R\'{e}nyi entropies\ \cite{RWISIT04}\ has also led to the development of
\textquotedblleft one-shot\textquotedblright\ information theory
\cite{R05,T12}, with applications to cryptography.

Part of what makes the R\'{e}nyi entropies so useful in applications is their
properties: convergence to the Shannon and relative entropies in the limit as
$\alpha\rightarrow1$, monotonicity in the parameter$~\alpha$, and additivity,
in addition to others. The convergence to the Shannon and relative entropies
ensures that, by taking this limit, one recovers asymptotic
information-theoretic statements, such as the data compression theorem or the
channel capacity theorem, from the more fine-grained statements. Monotonicity
in the parameter $\alpha$ ensures that $H_{\alpha}( p) $ gives more weight to
low surprisal events for $\alpha>1$ and vice versa for $\alpha<1$, helping to
characterize the aforementioned trade-off in information-theoretic settings.
The additivity property implies that the R\'{e}nyi entropies can simplify
immensely when evaluated for memoryless stochastic processes.

In light of the progress that the R\'{e}nyi paradigm has brought to
information theory, one is left to wonder if this could happen in more exotic
settings, such as quantum information theory and/or for \textquotedblleft
multipartite\textquotedblright\ settings (here by multipartite, we mean three
or more parties). This line of thought has led to the development of several
non-commutative generalizations of the R\'{e}nyi relative entropy in
(\ref{eq:Renyi-rel-ent}), which has in turn led to a sharpened understanding
of several quantum information-theoretic tasks (see \cite{CMW14,T15}\ and
references therein) and refinements of the uncertainty principle
\cite{CBTW15}. As far as we are aware, the development of the multipartite
generalization of the R\'{e}nyi entropy\ in (\ref{eq:Renyi-ent}) is less
explored, with the exception of a recent proposal \cite{BSW15a} for a
multipartite quantum generalization.

With the intent of developing either a multipartite classical or quantum
generalization of (\ref{eq:Renyi-ent}), one might suggest after a moment's
thought to replace a quantity which features a linear combination of entropies
by one with the same linear combination of R\'{e}nyi entropies. However, this
approach is objectively unsatisfactory in at least two regards:\ properties of
the original information measure are not preserved by doing so and one is not
guaranteed to have the powerful monotonicity in $\alpha$ property mentioned
above. For example, take the case of the conditional mutual information of a
tripartite density operator $\rho_{ABC}$ defined as%
\begin{equation}
I(A;B|C)_{\rho}\equiv H(AC)_{\rho}+H(BC)_{\rho}-H(C)_{\rho}-H(ABC)_{\rho
},\label{eq:orig-CMI}%
\end{equation}
where $H(F)_{\sigma}\equiv-\operatorname{Tr}\{\sigma_{F}\log\sigma_{F}\}$ is
the quantum entropy of a density operator $\sigma$ on system $F$. One of the
most important properties of this quantity is that it is non-negative (known
as strong subadditivity of quantum entropy \cite{PhysRevLett.30.434,LR73}),
and as a consequence, it is monotone non-increasing with respect to any
quantum channel applied to the system $A$ \cite{CW04} (by symmetry, the same
is true for one applied to $B$). However, if we define a R\'{e}nyi
generalization of $I(A;B|C)_{\rho}$ as $H_{\alpha}(AC)_{\rho}+H_{\alpha
}(BC)_{\rho}-H_{\alpha}(C)_{\rho}-H_{\alpha}(ABC)_{\rho}$, where $H_{\alpha
}(F)_{\sigma}\equiv\left[  \log\operatorname{Tr}\{\sigma_{F}^{\alpha
}\}\right]  /\left(  1-\alpha\right)  $, then explicit counterexamples reveal
that this R\'{e}nyi generalization can be\ negative, monotonicity with respect
to quantum channels need not hold, and neither does monotonicity in~$\alpha
$~\cite{LMW13}.

To remedy these deficiencies, the authors of \cite{BSW15a}\ put forward a
general prescription for producing a R\'{e}nyi generalization of a quantum
information measure, with the aim of having the properties of the original
measure retained while also satisfying the monotonicity in $\alpha$ property.
The work in \cite{BSW15a} was only partially successful in this regard.
Continuing with our example of conditional mutual information, consider the
following R\'{e}nyi generalization \cite{BSW14}:%
\begin{equation}
I_{\alpha}(A;B|C)_{\rho}\equiv\frac{1}{\alpha-1}\log\operatorname{Tr}\left\{
\rho_{ABC}^{\alpha}\rho_{AC}^{\left(  1-\alpha\right)  /2}\rho_{C}^{\left(
\alpha-1\right)  /2}\rho_{BC}^{1-\alpha}\rho_{C}^{\left(  \alpha-1\right)
/2}\rho_{AC}^{\left(  1-\alpha\right)  /2}\right\}  .\label{eq:Renyi-cmi-old}%
\end{equation}
For $\alpha\in\lbrack0,1)\cup(1,2]$, the quantity is non-negative, monotone
non-increasing with respect to quantum channels acting on the $B$ system,
converges to $I(A;B|C)_{\rho}$ in the limit as $\alpha\rightarrow1$, and is
conjectured to obey the monotonicity in $\alpha$ property (with some numerical
and analytical evidence in favor established)~\cite{BSW14}. However, hitherto
a proof of the monotonicity in $\alpha$ property for $I_{\alpha}(A;B|C)_{\rho
}$ remains lacking. It is also an open question to determine whether
$I_{\alpha}(A;B|C)_{\rho}$ is monotone non-increasing with respect to quantum
channels acting on the $A$ system---this partially has to do with the fact
that $I_{\alpha}(A;B|C)_{\rho}$ is not symmetric with respect to exchange of
the $A$ and $B$ systems, unlike the conditional mutual information in \eqref{eq:orig-CMI}.

\section{Summary of results}

\label{sec:summary-results}In this paper, we modify the recently proposed
R\'{e}nyi generalizations of quantum information measures from \cite{BSW15a}
by placing \textquotedblleft swivels\textquotedblright\ in a given chain of
operators.\footnote{A \textquotedblleft swivel\textquotedblright\ is a
coupling placed between two objects in a chain in order to allow for them to
\textquotedblleft swivel\textquotedblright\ about a given axis.} As an example
of the idea, consider that we can rewrite the quantity in
(\ref{eq:Renyi-cmi-old}) in terms of the Schatten 2-norm as follows:%
\begin{equation}
I_{\alpha}( A;B|C) _{\rho}\equiv\frac{2}{\alpha-1}\log\left\Vert \rho
_{BC}^{\left(  1-\alpha\right)  /2}\rho_{C}^{\left(  \alpha-1\right)  /2}%
\rho_{AC}^{\left(  1-\alpha\right)  /2}\rho_{ABC}^{\alpha/2}\right\Vert _{2}.
\end{equation}
The new idea is to modify this quantity to include swivels as follows:%
\begin{equation}
I_{\alpha}^{\prime}( A;B|C) _{\rho}\equiv\frac{2}{\alpha-1}\max_{V_{\rho_{AC}%
}\in\mathbb{V}_{\rho_{AC}},V_{\rho_{C}}\in\mathbb{V}_{\rho_{C}}}\log\left\Vert
\rho_{BC}^{\left(  1-\alpha\right)  /2}V_{\rho_{C}}\rho_{C}^{\left(
\alpha-1\right)  /2}\rho_{AC}^{\left(  1-\alpha\right)  /2}V_{\rho_{AC}}%
\rho_{ABC}^{\alpha/2}\right\Vert _{2}, \label{eq:swivel-CMI}%
\end{equation}
where $\mathbb{V}_{\omega}$ is the compact set of all unitaries $V_{\omega}%
$\ commuting with the Hermitian operator $\omega$. Thus, the fixed eigenbases
of $\rho_{C}$ and $\rho_{AC}$ act as swivels connecting adjacent operators in
the operator chain above, such that the unitary rotations $V_{\rho_{C}}$ and
$V_{\rho_{AC}}$ about these swivels are allowed. Of course, such swivels make
no difference when the density operator $\rho_{ABC}$ and its marginals commute
with each other (the classical case), or when the $C$ system is trivial, in
which case the above quantity reduces to a R\'{e}nyi mutual information%
\begin{align}
I_{\alpha}^{\prime}( A;B) _{\rho}  &  \equiv\frac{2}{\alpha-1}\max
_{V_{\rho_{A}}\in\mathbb{V}_{\rho_{A}}}\log\left\Vert \rho_{B}^{\left(
1-\alpha\right)  /2}\rho_{A}^{\left(  1-\alpha\right)  /2}V_{\rho_{A}}%
\rho_{AB}^{\alpha/2}\right\Vert _{2}\\
&  =\frac{2}{\alpha-1}\log\left\Vert \rho_{B}^{\left(  1-\alpha\right)
/2}\rho_{A}^{\left(  1-\alpha\right)  /2}\rho_{AB}^{\alpha/2}\right\Vert _{2}.
\end{align}
We mention that we were led to the definition in (\ref{eq:swivel-CMI}) as a
consequence of the developments in \cite{W15}, in which similar swivels
appeared in refinements of entropy inequalities such as monotonicity of
quantum relative entropy and strong subadditivity.

The quantity in (\ref{eq:swivel-CMI}) satisfies some of the properties already
established for $I_{\alpha}( A;B|C) _{\rho}$ in \cite{BSW14}, which include
non-negativity for $\alpha\in\lbrack0,1)\cup(1,2]$ and monotonicity with
respect to quantum channels acting on the $B$ system. However, the extra
degree of freedom in (\ref{eq:swivel-CMI})\ allows us to prove that this
swiveled R\'{e}nyi conditional mutual information is monotone non-decreasing
in $\alpha$ for $\alpha\in[ 0,1)\cup(1,2] $.

The swiveled R\'{e}nyi entropies are in general discontinuous at $\alpha=1$
and do not converge to the von Neumann entropy-based measures in the limit as
$\alpha\rightarrow1$. Thus, the present paper represents a work in progress
toward the general goal of find R\'{e}nyi generalizations of quantum
information measures that satisfy all of the desired properties that one would
like to have. It thus remains an open question to find R\'{e}nyi quantities
that meet all desiderata.

The rest of the paper proceeds by developing this idea in detail. We review
some background material in Section~\ref{sec:prelim}, which includes various
quantum R\'{e}nyi entropies and the Hadamard three-line theorem, the latter
being the essential tool for establishing monotonicity in $\alpha$ for the
swiveled R\'{e}nyi entropies. We then focus in Section~\ref{sec:swivel-Renyi}%
\ on developing swiveled R\'{e}nyi generalizations of the quantum relative
entropy difference in (\ref{eq:rel-ent-diff}), given that many different
information measures can be written in terms of this relative entropy
difference, including conditional mutual information (see, e.g., the
discussions in \cite{SBW14,W14,W15}). Our main contributions are
Theorems~\ref{thm:monotone}\ and \ref{thm:monotone-tilde}, which state that
these quantities are monotone non-decreasing in $\alpha$ for particular
values. We then briefly discuss how refinements of entropy inequalities follow
as a consequence of the properties of the swiveled R\'{e}nyi entropies.
Section~\ref{sec:Renyi-CMI}\ discusses swiveled R\'{e}nyi conditional mutual
informations and justifies that they possess the properties stated above. We
extend the idea in Section~\ref{sec:arbitrary-vn-measure}\ to establish
swiveled R\'{e}nyi generalizations of an arbitrary linear combination of von
Neumann entropies with coefficients chosen from the set $\left\{
-1,0,1\right\}  $. We finally show how our methods can be used to address an
open question posed in \cite{Z14b}. Section~\ref{sec:conclusion}\ concludes
with a summary and some open directions.

\section{Preliminaries}

\label{sec:prelim}

\subsection{Quantum states and channels}

A quantum state is described mathematically by a density operator, which is a
positive semi-definite operator with trace equal to one. A quantum channel is
a linear, trace-preserving, completely positive map. For more background on
quantum information theory, we refer to \cite{NC10,W13}. Our results apply to
finite-dimensional Hilbert spaces. For most developments, we take $\rho$,
$\sigma$, and $\mathcal{N}$ to be as given in the following definition:

\begin{definition}
\label{def:rho-sig-N}Let $\rho$ be a density operator acting on a
finite-dimensional Hilbert space $\mathcal{H}$, $\sigma$ be a non-zero
positive semi-definite operator acting on $\mathcal{H}$, and $\mathcal{N}$ be
a quantum channel, taking operators acting on $\mathcal{H}$ to those acting on
a finite-dimensional Hilbert space $\mathcal{K}$.
\end{definition}

\noindent Sometimes we need more restrictions, in which case we take $\rho$, $\sigma$,
and $\mathcal{N}$ as follows:

\begin{definition}
\label{def:rho-sig-N-PD} Let $\rho$, $\sigma$, and $\mathcal{N}$ be as given
in Definition~\ref{def:rho-sig-N}, with the additional restriction that $\rho$
and $\sigma$ are positive definite, and $\mathcal{N}$ is such that
$\mathcal{N}(\rho)$ and $\mathcal{N}(\sigma)$ are also positive definite.
\end{definition}

We employ the common convention that functions of Hermitian operators are
evaluated on their support. In more detail, the support of a Hermitian
operator $A$, written as $\operatorname{supp}(A)$, is defined as the vector
space spanned by its eigenvectors whose corresponding eigenvalues are
non-zero. Let an eigendecomposition of $A$ be given as $A = \sum_{i:a_{i}
\neq0} a_{i} \vert i\rangle\langle i \vert$ for eigenvectors $\{ \vert i
\rangle\}$. Then $\operatorname{supp}(A) = \operatorname{span} \{ \vert i
\rangle: a_{i} \neq0\}$. Let $\Pi_{A}$ denote the projection onto the support
of $A$. A function $f$ of an operator $A$ is then defined as $f(A) =
\sum_{i:a_{i}\neq0} f(a_{i}) \vert i\rangle\langle i \vert$.

\subsection{Entropies and norms}

Let $\rho$, $\sigma$, and $\mathcal{N}$ be as given in
Definition~\ref{def:rho-sig-N}. The quantum relative entropy \cite{U62}\ is
defined as%
\begin{equation}
D( \rho\Vert\sigma) \equiv\operatorname{Tr}\left\{  \rho\left[  \log\rho
-\log\sigma\right]  \right\}  ,
\end{equation}
whenever $\operatorname{supp}(\rho) \subseteq\operatorname{supp}(\sigma)$, and
otherwise, it is defined to be equal to $+\infty$. The quantum relative
entropy is monotone non-increasing with respect to quantum channels
\cite{Lindblad1975,U77}, in the sense that
\begin{equation}
D( \rho\Vert\sigma) \geq D( \mathcal{N}( \rho) \Vert\mathcal{N}( \sigma) ) .
\label{eq:mon-rel-ent}%
\end{equation}
Another relevant information measure is the quantum relative entropy
difference, defined as%
\begin{equation}
\Delta( \rho,\sigma,\mathcal{N}) \equiv D( \rho\Vert\sigma) -D( \mathcal{N}(
\rho) \Vert\mathcal{N}( \sigma) ) . \label{eq:rel-ent-diff}%
\end{equation}

We can use the Schatten norms in order to establish R\'{e}nyi\ generalizations
of von Neumann entropies, which are more refined information measures for
quantum states and channels that reduce to the von Neumann quantities in a
limit. The Schatten $p$-norm of an operator $A$ is defined as%
\begin{equation}
\left\Vert A\right\Vert _{p}\equiv\left[  \operatorname{Tr}\left\{  \left\vert
A\right\vert ^{p}\right\}  \right]  ^{1/p}, \label{eq:p-norm}%
\end{equation}
where $p\geq1$ and $\left\vert A\right\vert \equiv\sqrt{A^{\dag}A}$ (note that
we sometimes use the notation $\left\Vert A\right\Vert _{p}$ even for values
$p\in\left(  0,1\right)  $ when the quantity on the right-hand side\ of
(\ref{eq:p-norm}) is not a norm). From the above definition, we can see that
the following equalities hold for any operators $A$ and $B$:%
\begin{align}
\operatorname{Tr}\left\{  B^{\dag}A^{\dag}AB\right\}   &  =\left\Vert
AB\right\Vert _{2}^{2},\label{eq:helper-1}\\
\left\Vert B^{\dag}A^{\dag}AB\right\Vert _{p}^{p}  &  =\left\Vert
AB\right\Vert _{2p}^{2p}. \label{eq:helper-2}%
\end{align}

The quantum R\'{e}nyi entropy of a state $\rho$ is defined for $\alpha
\in\left(  0,1\right)  \cup\left(  1,\infty\right)  $ as%
\begin{equation}
H_{\alpha}( \rho) \equiv\frac{1}{1-\alpha}\log\operatorname{Tr} \{
\rho^{\alpha}\} = \frac{\alpha}{1-\alpha}\log\left\Vert \rho\right\Vert
_{\alpha},
\end{equation}
and reduces to the von Neumann entropy in the limit as $\alpha\rightarrow1$:%
\begin{equation}
\lim_{\alpha\rightarrow1}H_{\alpha}( \rho) =H( \rho) .
\end{equation}
There are at least two ways to generalize the quantum relative entropy,
%again
%defined for $\alpha\in\left(  0,1\right)  \cup\left(  1,\infty\right)  $,
which we refer to as the R\'{e}nyi relative entropy $D_{\alpha}( \rho
\Vert\sigma) $ \cite{P86} and the sandwiched R\'{e}nyi relative entropy
$\widetilde{D}_{\alpha}( \rho\Vert\sigma) $ \cite{MDSFT13,WWY13}. They are
defined respectively as follows:
\begin{align}
D_{\alpha}( \rho\Vert\sigma)  &  \equiv\frac{1}{\alpha-1}\log\operatorname{Tr}%
\left\{  \rho^{\alpha}\sigma^{1-\alpha}\right\} \\
&  =\frac{2}{\alpha-1}\log\left\Vert \sigma^{\left(  1-\alpha\right)  /2}%
\rho^{\alpha/2}\right\Vert _{2},\label{eq:D-rewrite-1}\\
\widetilde{D}_{\alpha}( \rho\Vert\sigma)  &  \equiv\frac{1}{\alpha-1}%
\log\operatorname{Tr}\left\{  \left(  \sigma^{\left(  1-\alpha\right)
/2\alpha}\rho\sigma^{\left(  1-\alpha\right)  /2\alpha}\right)  ^{\alpha
}\right\} \\
&  =\frac{2\alpha}{\alpha-1}\log\left\Vert \sigma^{\left(  1-\alpha\right)
/2\alpha}\rho^{1/2}\right\Vert _{2\alpha}, \label{eq:D-rewrite-2}%
\end{align}
if $\alpha\in(0,1)$ or if $\alpha\in(1,\infty)$ and $\operatorname{supp}(\rho)
\subseteq\operatorname{supp}(\sigma)$. If $\alpha\in(1,\infty)$ and
$\operatorname{supp}(\rho) \nsubseteq\operatorname{supp}(\sigma)$, then they
are defined to be equal to $+\infty$. The rewritings in (\ref{eq:D-rewrite-1})
and (\ref{eq:D-rewrite-2}) are helpful for our developments in this paper and
follow from (\ref{eq:helper-1})--(\ref{eq:helper-2}) and the following:%
\begin{align}
\operatorname{Tr}\left\{  \rho^{\alpha}\sigma^{1-\alpha}\right\}   &
=\operatorname{Tr}\left\{  \rho^{\alpha/2}\sigma^{\left(  1-\alpha\right)
/2}\sigma^{\left(  1-\alpha\right)  /2}\rho^{\alpha/2}\right\}  ,\\
\operatorname{Tr}\left\{  \left(  \sigma^{\left(  1-\alpha\right)  /2\alpha
}\rho\sigma_{\alpha}^{\left(  1-\alpha\right)  /2\alpha}\right)  ^{\alpha
}\right\}   &  =\left\Vert \sigma^{\left(  1-\alpha\right)  /2\alpha}%
\rho\sigma^{\left(  1-\alpha\right)  /2\alpha}\right\Vert _{\alpha}^{\alpha}\\
&  =\left\Vert \rho^{1/2}\sigma^{\left(  1-\alpha\right)  /\alpha}\rho
^{1/2}\right\Vert _{\alpha}^{\alpha}.
\end{align}
Both R\'{e}nyi generalizations reduce to the quantum relative entropy in the
limit as $\alpha\rightarrow1$ \cite{P86,MDSFT13,WWY13}:%
\begin{equation}
\lim_{\alpha\rightarrow1}D_{\alpha}( \rho\Vert\sigma) =\lim_{\alpha
\rightarrow1}\widetilde{D}_{\alpha}( \rho\Vert\sigma) =D( \rho\Vert\sigma) .
\end{equation}
The R\'{e}nyi relative entropy is monotone non-increasing with respect to
quantum channels when $\alpha\in[ 0,1)\cup(1,2] $ \cite{P86}:%
\begin{equation}
D_{\alpha}( \rho\Vert\sigma) \geq D_{\alpha}( \mathcal{N}( \rho)
\Vert\mathcal{N}( \sigma) ) ,
\end{equation}
and the sandwiched R\'{e}nyi relative entropy possesses a similar monotonicity
property when $\alpha\in[ 1/2,1)\cup(1,\infty) $ \cite{FL13,B13}:%
\begin{equation}
\widetilde{D}_{\alpha}( \rho\Vert\sigma) \geq\widetilde{D}_{\alpha}(
\mathcal{N}( \rho) \Vert\mathcal{N}( \sigma) ) .
\end{equation}

By picking particular values of the R\'enyi parameter $\alpha$, the quantities
above take on special forms and have meaning in operational contexts, being
known as the zero-relative entropy \cite{D09}, the collision relative entropy
\cite{DFW13}, the min-relative entropy \cite{DKFRR12}, and the max-relative
entropy \cite{D09}, respectively:%
\begin{align}
D_{0}( \rho\Vert\sigma)  &  =-\log\operatorname{Tr}\left\{  \rho^{0}%
\sigma\right\}  ,\\
D_{2}( \rho\Vert\sigma)  &  =\log\left\Vert \rho\sigma^{-1/2}\right\Vert
_{2},\\
\widetilde{D}_{1/2}( \rho\Vert\sigma)  &  =-\log F( \rho,\sigma) ,\\
D_{\max}( \rho\Vert\sigma)  &  =\lim_{\alpha\rightarrow\infty}\widetilde
{D}_{\alpha}( \rho\Vert\sigma) =\log\left\Vert \sigma^{-1/2}\rho\sigma
^{-1/2}\right\Vert _{\infty}=\log\left\Vert \sigma^{-1/2}\rho^{1/2}\right\Vert
_{\infty}^{2},
\end{align}
where $F( \rho,\sigma) \equiv\left\Vert \sqrt{\rho}\sqrt{\sigma}\right\Vert
_{1}^{2}$ is the quantum fidelity \cite{U73}.

\subsection{Hadamard three-line theorem}

One of the most important technical tools for proving our main result is the
operator version of the Hadamard three-line theorem given in \cite{B13}, in
particular, the very slight modification stated in \cite{D14}. We note that
the theorem below is a variant of the Riesz-Thorin operator interpolation
theorem (see, e.g., \cite{BL76,RS75}).

\begin{theorem}
\label{thm:hadamard}Let $S\equiv\left\{  z\in\mathbb{C}:0\leq\operatorname{Re}%
\left\{  z\right\}  \leq1\right\}  $, and let $L( \mathcal{H}) $ be the space
of bounded linear operators acting on a Hilbert space $\mathcal{H}$. Let
$G:S\rightarrow L( \mathcal{H}) $ be a bounded map that is holomorphic on the
interior of $S$ and continuous on the boundary.\footnote{A map $G:S\rightarrow
L(\mathcal{H})$ is holomorphic (continuous, bounded) if the corresponding
functions to matrix entries are holomorphic (continuous, bounded).} Let
$\theta\in\left(  0,1\right)  $ and define $p_{\theta}$ by%
\begin{equation}
\frac{1}{p_{\theta}}=\frac{1-\theta}{p_{0}}+\frac{\theta}{p_{1}},
\label{eq:p-relation}%
\end{equation}
where $p_{0},p_{1}\in\lbrack1,\infty]$. For $k=0,1$ define%
\begin{equation}
M_{k}=\sup_{t\in\mathbb{R}}\left\Vert G\left(  k+it\right)  \right\Vert
_{p_{k}}.
\end{equation}
Then%
\begin{equation}
\left\Vert G\left(  \theta\right)  \right\Vert _{p_{\theta}}\leq
M_{0}^{1-\theta}M_{1}^{\theta}. \label{eq:hadamard-3-line}%
\end{equation}

\end{theorem}

\subsection{R\'{e}nyi generalizations of the quantum relative entropy
difference}

Let $\rho$, $\sigma$, and $\mathcal{N}$ be as given in
Definition~\ref{def:rho-sig-N}. In \cite{SBW14}, two R\'{e}nyi generalizations
of the relative entropy difference in (\ref{eq:rel-ent-diff}) were defined as
follows:%
\begin{align}
\Delta_{\alpha}(\rho,\sigma,\mathcal{N})  &  \equiv\frac{1}{\alpha-1}%
\log\operatorname{Tr}\left\{  \rho^{\alpha}\sigma^{\left(  1-\alpha\right)
/2}\mathcal{N}^{\dag}\left(  \left[  \mathcal{N}(\sigma)\right]  ^{\left(
\alpha-1\right)  /2}\left[  \mathcal{N}(\rho)\right]  ^{1-\alpha}\left[
\mathcal{N}(\sigma)\right]  ^{\left(  \alpha-1\right)  /2}\right)
\sigma^{\left(  1-\alpha\right)  /2}\right\}  ,\nonumber\\
\widetilde{\Delta}_{\alpha}(\rho,\sigma,\mathcal{N})  &  \equiv\frac{1}%
{\alpha^{\prime}}\log\left\Vert \rho^{1/2}\sigma^{-\alpha^{\prime}%
/2}\mathcal{N}^{\dag}\left(  \left[  \mathcal{N}(\sigma)\right]
^{\alpha^{\prime}/2}\left[  \mathcal{N}(\rho)\right]  ^{-\alpha^{\prime}%
}\left[  \mathcal{N}(\sigma)\right]  ^{\alpha^{\prime}/2}\right)
\sigma^{-\alpha^{\prime}/2}\rho^{1/2}\right\Vert _{\alpha},
\end{align}
where $\alpha^{\prime}\equiv\left(  \alpha-1\right)  /\alpha$. Let $U$ be an
isometric extension of $\mathcal{N}$, so that%
\begin{equation}
\mathcal{N}(\cdot)=\operatorname{Tr}_{E}\left\{  U\left(  \cdot\right)
U^{\dag}\right\}  .
\end{equation}
We can write the adjoint $\mathcal{N}^{\dag}$\ in terms of this isometric
extension as follows:%
\begin{equation}
\mathcal{N}^{\dag}(\cdot)=U^{\dag}\left(  (\cdot)\otimes I_{E}\right)  U.
\end{equation}
This then allows us to write the definitions above in a simpler form:%
\begin{align}
\Delta_{\alpha}(\rho,\sigma,\mathcal{N})  &  =\frac{2}{\alpha-1}\log\left\Vert
\left(  \left[  \mathcal{N}(\rho)\right]  ^{\left(  1-\alpha\right)
/2}\left[  \mathcal{N}(\sigma)\right]  ^{\left(  \alpha-1\right)  /2}\otimes
I_{E}\right)  U\sigma^{\left(  1-\alpha\right)  /2}\rho^{\alpha/2}\right\Vert
_{2},\\
\widetilde{\Delta}_{\alpha}(\rho,\sigma,\mathcal{N})  &  =\frac{2}%
{\alpha^{\prime}}\log\left\Vert \left(  \left[  \mathcal{N}(\rho)\right]
^{-\alpha^{\prime}/2}\left[  \mathcal{N}(\sigma)\right]  ^{\alpha^{\prime}%
/2}\otimes I_{E}\right)  U\sigma^{-\alpha^{\prime}/2}\rho^{1/2}\right\Vert
_{2\alpha}.
\end{align}
It is known that the following limits hold for $\rho$, $\sigma$, and
$\mathcal{N}$ taken as in Definition \ref{def:rho-sig-N-PD} \cite{SBW14}:%
\begin{equation}
\lim_{\alpha\rightarrow1}\Delta_{\alpha}(\rho,\sigma,\mathcal{N})=\lim
_{\alpha\rightarrow1}\widetilde{\Delta}_{\alpha}(\rho,\sigma,\mathcal{N}%
)=\Delta(\rho,\sigma,\mathcal{N}).
\end{equation}
The fact that these limits hold for $\rho$, $\sigma$, and $\mathcal{N}$ taken
as in Definition \ref{def:rho-sig-N} and subject to $\operatorname{supp}%
(\rho)\subseteq\operatorname{supp}(\sigma)$ follows from \cite{W15} and the
development in Appendix~\ref{app:Delta-limit-alpha-1}. \cite{DW15} proved that
for $\alpha\in\lbrack0,1)\cup(1,2]$,%
\begin{equation}
\Delta_{\alpha}(\rho,\sigma,\mathcal{N})\geq0,
\end{equation}
and for $\alpha\in\lbrack1/2,1)\cup(1,\infty]$:%
\begin{equation}
\widetilde{\Delta}_{\alpha}(\rho,\sigma,\mathcal{N})\geq0,
\end{equation}
when $\rho$, $\sigma$, and $\mathcal{N}$ are taken as in Definition
\ref{def:rho-sig-N-PD}. The latter inequality was refined recently in
\cite{W15} for $\alpha\in(1/2,1]$ and for $\rho$, $\sigma$, and $\mathcal{N}$
taken as in Definition \ref{def:rho-sig-N} and subject to $\operatorname{supp}%
(\rho)\subseteq\operatorname{supp}(\sigma)$. It remains an open question to
determine whether these quantities are non-decreasing in $\alpha$ for any
non-trivial range of $\alpha$ (note that \cite{SBW14} argued that they are
non-decreasing in $\alpha$ in a neighborhood of $\alpha=1$).

\section{Swiveled R\'{e}nyi generalizations of the quantum relative entropy
difference}

\label{sec:swivel-Renyi}In the spirit of the discussion in
Section~\ref{sec:summary-results}, we consider different definitions of
$\Delta_{\alpha}( \rho,\sigma,\mathcal{N}) $ and $\widetilde{\Delta}_{\alpha}(
\rho,\sigma,\mathcal{N}) $\ in order to allow for unitary rotations about
swivels, i.e., an optimization over unitaries of the form $V_{\mathcal{N}(
\sigma) }$ and $V_{\sigma}$:

\begin{definition}
Let $\rho$, $\sigma$, and $\mathcal{N}$ be as given in
Definition~\ref{def:rho-sig-N}. We define swiveled R\'{e}nyi generalizations
of the quantum relative entropy difference in \eqref{eq:rel-ent-diff} as
follows:%
\begin{align}
\Delta_{\alpha}^{\prime}(\rho,\sigma,\mathcal{N})  &  \equiv\frac{2}{\alpha
-1}\max_{V_{\sigma},V_{\mathcal{N}(\sigma)}}\log\left\Vert \left(  \left[
\mathcal{N}(\rho)\right]  ^{\left(  1-\alpha\right)  /2}V_{\mathcal{N}%
(\sigma)}\left[  \mathcal{N}(\sigma)\right]  ^{\left(  \alpha-1\right)
/2}\otimes I_{E}\right)  U\sigma^{\left(  1-\alpha\right)  /2}V_{\sigma}%
\rho^{\alpha/2}\right\Vert _{2},\label{eq:Delta-new}\\
\widetilde{\Delta}_{\alpha}^{\prime}(\rho,\sigma,\mathcal{N})  &  \equiv
\frac{2}{\alpha^{\prime}}\max_{V_{\sigma},V_{\mathcal{N}(\sigma)}}%
\log\left\Vert \left(  \left[  \mathcal{N}(\rho)\right]  ^{-\alpha^{\prime}%
/2}V_{\mathcal{N}(\sigma)}\left[  \mathcal{N}(\sigma)\right]  ^{\alpha
^{\prime}/2}\otimes I_{E}\right)  U\sigma^{-\alpha^{\prime}/2}V_{\sigma}%
\rho^{1/2}\right\Vert _{2\alpha}, \label{eq:Delta-tilde-new}%
\end{align}
where $\alpha^{\prime}=\left(  \alpha-1\right)  /\alpha$ and the optimizations
are over the compact sets of unitaries $V_{\sigma}$ and $V_{\mathcal{N}%
(\sigma)}$ commuting with $\sigma$ and $\mathcal{N}(\sigma)$, respectively.
\end{definition}

This slight extra degree of freedom allows us to establish that $\Delta
_{\alpha}^{\prime}$ and $\widetilde{\Delta}_{\alpha}^{\prime}$ are monotone
non-decreasing in $\alpha$ for particular values (see
Theorems~\ref{thm:monotone} and \ref{thm:monotone-tilde}).

\subsection{Reduction to R\'enyi relative entropy}

Observe that by choosing $\mathcal{N}=\operatorname{Tr}$, we find that
$\Delta_{\alpha}^{\prime}$ reduces to the R\'{e}nyi relative entropy whenever
$\operatorname{supp}(\rho)\subseteq\operatorname{supp}(\sigma)$:%
\begin{align}
\Delta_{\alpha}^{\prime}(\rho,\sigma,\operatorname{Tr})  &  =\frac{2}%
{\alpha-1}\log\left\Vert \sigma^{\left(  1-\alpha\right)  /2}\rho^{\alpha
/2}\right\Vert _{2}+\log\operatorname{Tr}\left\{  \sigma\right\} \\
&  =D_{\alpha}(\rho\Vert\sigma)+\log\operatorname{Tr}\left\{  \sigma\right\}
,
\end{align}
and $\widetilde{\Delta}_{\alpha}^{\prime}$ to the sandwiched R\'{e}nyi
relative entropy whenever $\operatorname{supp}(\rho)\subseteq
\operatorname{supp}(\sigma)$:%
\begin{align}
\widetilde{\Delta}_{\alpha}^{\prime}\left(  \rho,\sigma,\operatorname{Tr}%
\right)   &  \equiv\frac{2}{\alpha^{\prime}}\log\left\Vert \sigma
^{-\alpha^{\prime}/2}\rho^{1/2}\right\Vert _{2\alpha}+\log\operatorname{Tr}%
\left\{  \sigma\right\} \\
&  =\widetilde{D}_{\alpha}(\rho\Vert\sigma)+\log\operatorname{Tr}\left\{
\sigma\right\}  ,
\end{align}
just as%
\begin{equation}
\Delta(\rho,\sigma,\operatorname{Tr})=D(\rho\Vert\sigma)+\log\operatorname{Tr}%
\left\{  \sigma\right\}  .
\end{equation}

\subsection{Behavior around $\alpha=1$}

Here we discuss the behavior of $\Delta_{\alpha}^{\prime}$ and $\widetilde
{\Delta}_{\alpha}^{\prime}$ around $\alpha=1$, with the result being that
these quantities are generally discontinuous at $\alpha=1$:

\begin{proposition}
\label{prop:lim-a-1}Let $\rho$, $\sigma$, and $\mathcal{N}$ be as given in
Definition~\ref{def:rho-sig-N-PD}. Then%
\begin{align}
\lim_{\alpha\nearrow1}\Delta_{\alpha}^{\prime}(\rho,\sigma,\mathcal{N})  &
=\lim_{\alpha\nearrow1}\widetilde{\Delta}_{\alpha}^{\prime}(\rho
,\sigma,\mathcal{N})=\min_{V_{\mathcal{N}(\sigma)},V_{\sigma}}%
f(1,V_{\mathcal{N}(\sigma)},V_{\sigma}),\label{eq:inequalities-disc}\\
\lim_{\alpha\searrow1}\Delta_{\alpha}^{\prime}(\rho,\sigma,\mathcal{N})  &
=\lim_{\alpha\searrow1}\widetilde{\Delta}_{\alpha}^{\prime}(\rho
,\sigma,\mathcal{N})=\max_{V_{\mathcal{N}(\sigma)},V_{\sigma}}%
f(1,V_{\mathcal{N}(\sigma)},V_{\sigma}), \label{eq:inequalities-disc-2}%
\end{align}
where%
\begin{multline}
f(1,V_{\mathcal{N}(\sigma)},V_{\sigma})\equiv\operatorname{Tr}\left\{
\rho\left[  \log\rho-\log\sigma\right]  \right\} \label{eq:f-function}\\
-\operatorname{Tr}\left\{  \mathcal{N}\left(  \left[  V_{\sigma}\rho
V_{\sigma}^{\dag}\right]  \right)  \left[  \log\left[  V_{\mathcal{N}(\sigma
)}^{\dag}\mathcal{N}(\rho)V_{\mathcal{N}(\sigma)}\right]  -\log\left[
\mathcal{N}(\sigma)\right]  \right]  \right\}  .
\end{multline}
As a consequence, we have that%
\begin{equation}
\min_{V_{\mathcal{N}(\sigma)},V_{\sigma}}f(1,V_{\mathcal{N}(\sigma)}%
,V_{\sigma})\leq f\left(  1,I,I\right)  =\Delta(\rho,\sigma,\mathcal{N}%
)\leq\max_{V_{\mathcal{N}(\sigma)},V_{\sigma}}f(1,V_{\mathcal{N}(\sigma
)},V_{\sigma}),
\end{equation}
and there is generally a discontinuity at $\alpha=1$.
\end{proposition}

\begin{proof}
Let $\mathcal{A}\subseteq\lbrack0,2]$, which we will choose shortly. Define
the function $f:\mathcal{A}\times\mathbb{V}_{\mathcal{N}(\sigma)}%
\times\mathbb{V}_{\sigma}\rightarrow\mathbb{R}$ as%
\begin{equation}
f(\alpha,V_{\mathcal{N}(\sigma)},V_{\sigma})\equiv\frac{2}{\alpha-1}%
\log\left\Vert \left(  \left[  \mathcal{N}(\rho)\right]  ^{\left(
1-\alpha\right)  /2}V_{\mathcal{N}(\sigma)}\left[  \mathcal{N}(\sigma)\right]
^{\left(  \alpha-1\right)  /2}\otimes I_{E}\right)  U\sigma^{\left(
1-\alpha\right)  /2}V_{\sigma}\rho^{\alpha/2}\right\Vert _{2},
\label{eq:f-alpha-function}%
\end{equation}
whenever $\alpha\neq1$, and $f(1,V_{\mathcal{N}(\sigma)},V_{\sigma})$ as in
(\ref{eq:f-function}). One can check that%
\begin{equation}
\lim_{\alpha\rightarrow1}f(\alpha,V_{\mathcal{N}(\sigma)},V_{\sigma
})=f(1,V_{\mathcal{N}(\sigma)},V_{\sigma}),
\end{equation}
for example by performing Taylor expansions to calculate the limit (see
Appendix~\ref{app:taylor}\ for details of this calculation). The function $f$
is then continuous in $\alpha$, $V_{\sigma}$, and $V_{\mathcal{N}(\sigma)}$.
Furthermore, it fulfills the conditions of Lemma~\ref{lem:max-continuous} in
Appendix~\ref{app:auxiliary} if we choose $\mathcal{A}=[1,M]$ for any
$M\in(1,2]$ and $\mathcal{T}=\mathbb{V}_{\mathcal{N}(\sigma)}\times
\mathbb{V}_{\sigma}$. Hence, we get that%
\begin{equation}
\Delta_{\alpha}^{\prime}(\rho,\sigma,\mathcal{N})=\max_{V_{\mathcal{N}%
(\sigma)},V_{\sigma}}f(\alpha,V_{\mathcal{N}(\sigma)},V_{\sigma})
\end{equation}
is continuous on $\alpha\in\lbrack1,M]$ and thus%
\begin{equation}
\lim_{\alpha\searrow1}\Delta_{\alpha}^{\prime}(\rho,\sigma,\mathcal{N}%
)=\max_{V_{\mathcal{N}(\sigma)},V_{\sigma}}f(1,V_{\mathcal{N}(\sigma
)},V_{\sigma}).
\end{equation}
Repeating the same argument with $\mathcal{A}=[0,1]$ yields that%
\begin{equation}
\Delta_{\alpha}^{\prime}(\rho,\sigma,\mathcal{N})=\min_{V_{\mathcal{N}%
(\sigma)},V_{\sigma}}f(\alpha,V_{\mathcal{N}(\sigma)},V_{\sigma})
\end{equation}
is continuous on $[0,1]$ and thus%
\begin{equation}
\lim_{\alpha\nearrow1}\Delta_{\alpha}^{\prime}(\rho,\sigma,\mathcal{N}%
)=\min_{V_{\mathcal{N}(\sigma)},V_{\sigma}}f(1,V_{\mathcal{N}(\sigma
)},V_{\sigma}).
\end{equation}
Given that $\Delta(\rho,\sigma,\mathcal{N})=f(1,I,I)$, we can conclude the
following inequality:%
\begin{equation}
\min_{V_{\mathcal{N}(\sigma)},V_{\sigma}}f(1,V_{\mathcal{N}(\sigma)}%
,V_{\sigma})\leq\Delta(\rho,\sigma,\mathcal{N})\leq\max_{V_{\mathcal{N}%
(\sigma)},V_{\sigma}}f(1,V_{\mathcal{N}(\sigma)},V_{\sigma})
\end{equation}

The arguments for the quantity $\widetilde{\Delta}_{\alpha}^{\prime}%
(\rho,\sigma,\mathcal{N})$ are similar, so we just sketch them briefly. Define
the function%
\begin{equation}
g(\alpha,V_{\mathcal{N}(\sigma)},V_{\sigma})\equiv\frac{2\alpha}{\alpha-1}%
\log\left\Vert \left(  \left[  \mathcal{N}(\rho)\right]  ^{\left(
1-\alpha\right)  /2\alpha}V_{\mathcal{N}(\sigma)}\left[  \mathcal{N}%
(\sigma)\right]  ^{\left(  \alpha-1\right)  /2\alpha}\otimes I_{E}\right)
U\sigma^{\left(  1-\alpha\right)  /2\alpha}V_{\sigma}\rho^{1/2}\right\Vert
_{2\alpha},
\end{equation}
for $\alpha\neq1$ and set $g(1,V_{\mathcal{N}(\sigma)},V_{\sigma
})=f(1,V_{\mathcal{N}(\sigma)},V_{\sigma})$. One can then compute (again via
Taylor expansions, e.g.) that%
\begin{equation}
\lim_{\alpha\rightarrow1}g(\alpha,V_{\mathcal{N}(\sigma)},V_{\sigma
})=g(1,V_{\mathcal{N}(\sigma)},V_{\sigma}). \label{eq:g-limit-1}%
\end{equation}
The rest of the argument proceeds as above, which leads to the other
equalities in (\ref{eq:inequalities-disc})-(\ref{eq:inequalities-disc-2}).
\end{proof}

\subsection{Monotonicity in the R\'{e}nyi parameter}

This section contains our main result, that both $\Delta_{\alpha}^{\prime}$
and $\widetilde{\Delta}_{\alpha}^{\prime}$ are monotone non-decreasing with
respect to $\alpha$ for particular values.

\begin{theorem}
\label{thm:monotone}Let $\rho$, $\sigma$, and $\mathcal{N}$ be as given in
Definition~\ref{def:rho-sig-N}. The swiveled R\'{e}nyi quantity $\Delta
_{\alpha}^{\prime}( \rho,\sigma,\mathcal{N}) $ is monotone non-decreasing with
respect to $\alpha\in\left[  0,1)\cup(1,2\right]  $,\ in the sense that for
$0\leq\alpha\leq\gamma\leq2$, $\alpha\neq1$, and $\gamma\neq1$%
\begin{equation}
\Delta_{\alpha}^{\prime}( \rho,\sigma,\mathcal{N}) \leq\Delta_{\gamma}%
^{\prime}( \rho,\sigma,\mathcal{N}) . \label{eq:Delta-monotone}%
\end{equation}

\end{theorem}

\begin{proof}
The main tool for our proof is Theorem~\ref{thm:hadamard}. We break the proof
of inequality in (\ref{eq:Delta-monotone}) into several cases. We first
consider $1<\alpha<\gamma\leq2$. For some $W_{\mathcal{N}(\sigma)}%
\in\mathbb{V}_{\mathcal{N}(\sigma)}$ and $W_{\sigma}\in\mathbb{V}_{\sigma}$,
pick%
\begin{align}
G\left(  z\right)   &  =\left[  \mathcal{N}(\rho)\right]  ^{-z\left(
\gamma-1\right)  /2}W_{\mathcal{N}(\sigma)}\left[  \mathcal{N}(\sigma)\right]
^{z\left(  \gamma-1\right)  /2}U\sigma^{-z\left(  \gamma-1\right)
/2}W_{\sigma}\rho^{\left(  1+z\left(  \gamma-1\right)  \right)  /2}%
,\label{eq:first-G}\\
p_{0}  &  =2,\\
p_{1}  &  =2,\\
\theta &  =\frac{\alpha-1}{\gamma-1}\in\left(  0,1\right)  ,
\label{eq:first-theta}%
\end{align}
which fixes $p_{\theta}=2$. Then%
\begin{align}
M_{0}  &  =\sup_{t\in\mathbb{R}}\left\Vert G\left(  it\right)  \right\Vert
_{2}\label{eq:first-chain-1}\\
&  =\sup_{t\in\mathbb{R}}\left\Vert \left[  \mathcal{N}(\rho)\right]
^{-it\left(  \gamma-1\right)  /2}W_{\mathcal{N}(\sigma)}\left[  \mathcal{N}%
(\sigma)\right]  ^{it\left(  \gamma-1\right)  /2}U\sigma^{-it\left(
\gamma-1\right)  /2}W_{\sigma}\rho^{\left(  1+it\left(  \gamma-1\right)
\right)  /2}\right\Vert _{2}\\
&  =\left\Vert \rho^{1/2}\right\Vert _{2}=1,\\
M_{1}  &  =\sup_{t\in\mathbb{R}}\left\Vert G\left(  1+it\right)  \right\Vert
_{2}\\
&  =\sup_{t\in\mathbb{R}}\left\Vert \left[  \mathcal{N}(\rho)\right]
^{-\frac{\left(  1+it\right)  }{2}\left(  \gamma-1\right)  }W_{\mathcal{N}%
(\sigma)}\left[  \mathcal{N}(\sigma)\right]  ^{\frac{\left(  1+it\right)  }%
{2}\left(  \gamma-1\right)  }U\sigma^{-\frac{\left(  1+it\right)  }{2}\left(
\gamma-1\right)  }W_{\sigma}\rho^{\frac{\left(  1+\left(  1+it\right)  \left(
\gamma-1\right)  \right)  }{2}}\right\Vert _{2}\\
&  \leq\max_{V_{\mathcal{N}(\sigma)},V_{\sigma}}\left\Vert \left[
\mathcal{N}(\rho)\right]  ^{\left(  1-\gamma\right)  /2}V_{\mathcal{N}%
(\sigma)}\left[  \mathcal{N}(\sigma)\right]  ^{\left(  \gamma-1\right)
/2}U\sigma^{\left(  1-\gamma\right)  /2}V_{\sigma}\rho^{\gamma/2}\right\Vert
_{2}\\
&  =\exp\left\{  \frac{\gamma-1}{2}\Delta_{\gamma}^{\prime}(\rho
,\sigma,\mathcal{N})\right\}  ,\\
\left\Vert G\left(  \theta\right)  \right\Vert _{2}  &  =\left\Vert \left[
\mathcal{N}(\rho)\right]  ^{\left(  1-\alpha\right)  /2}W_{\mathcal{N}%
(\sigma)}\left[  \mathcal{N}(\sigma)\right]  ^{\left(  \alpha-1\right)
/2}U\sigma^{\left(  1-\alpha\right)  /2}W_{\sigma}\rho^{\alpha/2}\right\Vert
_{2}.
\end{align}
We then apply Theorem~\ref{thm:hadamard}\ to find that the following
inequality holds for all $W_{\mathcal{N}(\sigma)}\in\mathbb{V}_{\mathcal{N}%
(\sigma)}$ and $W_{\sigma}\in\mathbb{V}_{\sigma}$:%
\begin{equation}
\left\Vert \left[  \mathcal{N}(\rho)\right]  ^{\left(  1-\alpha\right)
/2}W_{\mathcal{N}(\sigma)}\left[  \mathcal{N}(\sigma)\right]  ^{\left(
\alpha-1\right)  /2}U\sigma^{\left(  1-\alpha\right)  /2}W_{\sigma}%
\rho^{\alpha/2}\right\Vert _{2}\leq\left[  \exp\left\{  \frac{\gamma-1}%
{2}\Delta_{\gamma}^{\prime}(\rho,\sigma,\mathcal{N})\right\}  \right]
^{\frac{\alpha-1}{\gamma-1}}. \label{eq:first-chain-last}%
\end{equation}
As a consequence, we can take the maximum over all $W_{\mathcal{N}(\sigma)}%
\in\mathbb{V}_{\mathcal{N}(\sigma)}$ and $W_{\sigma}\in\mathbb{V}_{\sigma}$
and apply the definition in (\ref{eq:Delta-new}) to establish that%
\begin{equation}
\exp\left\{  \frac{\alpha-1}{2}\Delta_{\alpha}^{\prime}(\rho,\sigma
,\mathcal{N})\right\}  \leq\left[  \exp\left\{  \frac{\gamma-1}{2}%
\Delta_{\gamma}^{\prime}(\rho,\sigma,\mathcal{N})\right\}  \right]
^{\frac{\alpha-1}{\gamma-1}}.
\end{equation}
We finally apply a logarithm to arrive at the conclusion that
(\ref{eq:Delta-monotone}) holds for all $1<\alpha<\gamma\leq2$.

To get the monotonicity for the range $0\leq\alpha<\gamma<1$, we exchange
$\alpha$ and $\gamma$ in (\ref{eq:first-G})-(\ref{eq:first-theta}) and apply
the same reasoning as in (\ref{eq:first-chain-1})-(\ref{eq:first-chain-last})
to arrive at the following inequality:%
\begin{equation}
\exp\left\{  \frac{\gamma-1}{2}\Delta_{\gamma}^{\prime}(\rho,\sigma
,\mathcal{N})\right\}  \leq\left[  \exp\left\{  \frac{\alpha-1}{2}%
\Delta_{\alpha}^{\prime}(\rho,\sigma,\mathcal{N})\right\}  \right]
^{\frac{\gamma-1}{\alpha-1}}.
\end{equation}
Taking a negative logarithm and noting that $0\leq\alpha<\gamma<1$ then gives
(\ref{eq:Delta-monotone}) for this range.

We are now left with proving the case $\alpha\in\lbrack0,1)$ and $\gamma
\in(1,2]$ the dual parameter of $\alpha$, such that $\alpha+\gamma=2$. Notice
that $\alpha-1=-\left(  \gamma-1\right)  $. Let $f\left(  z,\gamma\right)
=\left(  1-2z\right)  \left(  \gamma-1\right)  $. We pick%
\begin{align}
G\left(  z\right)   &  =\left[  \mathcal{N}(\rho)\right]  ^{-f\left(
z,\gamma\right)  /2}\left[  \mathcal{N}(\sigma)\right]  ^{f\left(
z,\gamma\right)  /2}U\sigma^{-f\left(  z,\gamma\right)  /2}\rho^{\left(
1+f\left(  z,\gamma\right)  \right)  /2},\\
p_{0}  &  =2,\\
p_{1}  &  =2,\\
\theta &  =1/2,
\end{align}
so that $p_{\theta}=2$. Consider that $f\left(  \theta,\gamma\right)  =0$, so
that%
\begin{align}
\left\Vert G\left(  \theta\right)  \right\Vert _{2}  &  =\left\Vert \left[
\mathcal{N}(\rho)\right]  ^{-f\left(  \theta,\gamma\right)  /2}\left[
\mathcal{N}(\sigma)\right]  ^{f\left(  \theta,\gamma\right)  /2}%
U\sigma^{-f\left(  \theta,\gamma\right)  /2}\rho^{\left(  1+f\left(
\theta,\gamma\right)  \right)  /2}\right\Vert _{2}\\
&  =\left\Vert U\rho^{1/2}\right\Vert _{2}=\left\Vert \rho^{1/2}\right\Vert
_{2}=1.
\end{align}
We then find that%
\begin{align}
M_{0}  &  =\sup_{t\in\mathbb{R}}\left\Vert G\left(  it\right)  \right\Vert
_{2}\\
&  =\sup_{t\in\mathbb{R}}\left\Vert \left[  \mathcal{N}(\rho)\right]
^{-\left(  1-2it\right)  \left(  \gamma-1\right)  /2}\left[  \mathcal{N}%
(\sigma)\right]  ^{\left(  1-2it\right)  \left(  \gamma-1\right)  /2}%
U\sigma^{-\left(  1-2it\right)  \left(  \gamma-1\right)  /2}\rho^{\left(
1+\left(  1-2it\right)  \left(  \gamma-1\right)  \right)  /2}\right\Vert
_{2}\\
&  \leq\max_{V_{\mathcal{N}(\sigma)},V_{\sigma}}\left\Vert \left[
\mathcal{N}(\rho)\right]  ^{\left(  1-\gamma\right)  /2}V_{\mathcal{N}%
(\sigma)}\left[  \mathcal{N}(\sigma)\right]  ^{\left(  \gamma-1\right)
/2}U\sigma^{\left(  1-\gamma\right)  /2}V_{\sigma}\rho^{\gamma/2}\right\Vert
_{2}\\
&  =\exp\left\{  \frac{\gamma-1}{2}\Delta_{\gamma}^{\prime}(\rho
,\sigma,\mathcal{N})\right\}  .
\end{align}
Consider that%
\begin{equation}
f\left(  1+it,\gamma\right)  =\left(  1-2\left(  1+it\right)  \right)  \left(
\gamma-1\right)  =-\left(  1+2it\right)  \left(  \gamma-1\right)  =\left(
1+2it\right)  \left(  \alpha-1\right)  .
\end{equation}
Thus, similarly, we have%
\begin{align}
M_{1}  &  =\sup_{t\in\mathbb{R}}\left\Vert G\left(  1+it\right)  \right\Vert
_{2}\\
&  =\sup_{t\in\mathbb{R}}\left\Vert \left[  \mathcal{N}(\rho)\right]
^{-\left(  1+2it\right)  \left(  \alpha-1\right)  /2}\left[  \mathcal{N}%
(\sigma)\right]  ^{\left(  1+2it\right)  \left(  \alpha-1\right)  /2}%
U\sigma^{-\left(  1+2it\right)  \left(  \alpha-1\right)  /2}\rho^{\left(
1+\left(  1+2it\right)  \left(  \alpha-1\right)  \right)  /2}\right\Vert
_{2}\\
&  \leq\max_{V_{\mathcal{N}(\sigma)},V_{\sigma}}\left\Vert \left[
\mathcal{N}(\rho)\right]  ^{\left(  1-\alpha\right)  /2}V_{\mathcal{N}%
(\sigma)}\left[  \mathcal{N}(\sigma)\right]  ^{\left(  \alpha-1\right)
/2}U\sigma^{\left(  1-\alpha\right)  /2}V_{\sigma}\rho^{\alpha/2}\right\Vert
_{2}\\
&  =\exp\left\{  \frac{\alpha-1}{2}\Delta_{\alpha}^{\prime}(\rho
,\sigma,\mathcal{N})\right\}  .
\end{align}
Applying Theorem~\ref{thm:hadamard} gives%
\begin{align}
1  &  \leq\exp\left\{  \frac{\gamma-1}{4}\Delta_{\gamma}^{\prime}(\rho
,\sigma,\mathcal{N})\right\}  \exp\left\{  \frac{\alpha-1}{4}\Delta_{\alpha
}^{\prime}(\rho,\sigma,\mathcal{N})\right\} \\
&  =\exp\left\{  \frac{\gamma-1}{4}\Delta_{\gamma}^{\prime}(\rho
,\sigma,\mathcal{N})\right\}  \exp\left\{  \frac{-\left(  \gamma-1\right)
}{4}\Delta_{\alpha}^{\prime}(\rho,\sigma,\mathcal{N})\right\}  ,
\end{align}
which implies (\ref{eq:Delta-monotone}) for $\alpha\in\lbrack0,1)$ and
$\gamma=2-\alpha$. Putting the three cases together along with
Proposition~\ref{prop:lim-a-1}\ gives the inequality in
(\ref{eq:Delta-monotone}) for $0\leq\alpha\leq\gamma\leq2$, $\alpha\neq1$, and
$\gamma\neq1$.
\end{proof}

\bigskip

\begin{theorem}
\label{thm:monotone-tilde}Let $\rho$, $\sigma$, and $\mathcal{N}$ be as given
in Definition~\ref{def:rho-sig-N}. The swiveled R\'{e}nyi quantity
$\widetilde{\Delta}_{\alpha}^{\prime}(\rho,\sigma,\mathcal{N})$ is monotone
non-decreasing with respect to $\alpha\in\lbrack1/2,1)\cup(1,\infty]$, in the
sense that for $1/2\leq\alpha\leq\gamma\leq\infty$, $\alpha\neq1$, and
$\gamma\neq1$%
\begin{equation}
\widetilde{\Delta}_{\alpha}^{\prime}(\rho,\sigma,\mathcal{N})\leq
\widetilde{\Delta}_{\gamma}^{\prime}(\rho,\sigma,\mathcal{N}).
\label{eq:Delta-tilde-monotone}%
\end{equation}

\end{theorem}

\begin{proof}
We handle the inequality in (\ref{eq:Delta-tilde-monotone}) in a similar way
as in the previous proof. First, suppose that $1<\alpha<\gamma$. Let
$\alpha^{\prime}=\left(  \alpha-1\right)  /\alpha$ and $\gamma^{\prime
}=\left(  \gamma-1\right)  /\gamma$, and note that $\alpha^{\prime}%
,\gamma^{\prime}>0$ for the choices given. For some $W_{\mathcal{N}(\sigma
)}\in\mathbb{V}_{\mathcal{N}(\sigma)}$ and $W_{\sigma}\in\mathbb{V}_{\sigma}$,
pick%
\begin{align}
G\left(  z\right)   &  =\left[  \mathcal{N}(\rho)\right]  ^{-z\gamma^{\prime
}/2}W_{\mathcal{N}(\sigma)}\left[  \mathcal{N}(\sigma)\right]  ^{z\gamma
^{\prime}/2}U\sigma^{-z\gamma^{\prime}/2}W_{\sigma}\rho^{1/2}%
,\label{eq:other-G}\\
p_{0}  &  =2,\\
p_{1}  &  =2\gamma,\\
\theta &  =\frac{\alpha^{\prime}}{\gamma^{\prime}}\in\left(  0,1\right)  ,
\label{eq:other-theta}%
\end{align}
which fixes $p_{\theta}=2\alpha$. Then we find the following expression for
$M_{0}$%
\begin{align}
M_{0}  &  =\sup_{t\in\mathbb{R}}\left\Vert G\left(  it\right)  \right\Vert
_{2}\label{eq:other-chain-1}\\
&  =\sup_{t\in\mathbb{R}}\left\Vert \left[  \mathcal{N}(\rho)\right]
^{-it\gamma^{\prime}/2}W_{\mathcal{N}(\sigma)}\left[  \mathcal{N}%
(\sigma)\right]  ^{it\gamma^{\prime}/2}U\sigma^{-it\gamma^{\prime}/2}%
W_{\sigma}\rho^{1/2}\right\Vert _{2}\\
&  =\left\Vert \rho^{1/2}\right\Vert _{2}=1,
\end{align}
and the following ones for $M_{1}$ and $\left\Vert G\left(  \theta\right)
\right\Vert _{2\alpha}$:%
\begin{align}
M_{1}  &  =\sup_{t\in\mathbb{R}}\left\Vert G\left(  1+it\right)  \right\Vert
_{2\gamma}\\
&  =\sup_{t\in\mathbb{R}}\left\Vert \left[  \mathcal{N}(\rho)\right]
^{-\left(  1+it\right)  \gamma^{\prime}/2}W_{\mathcal{N}(\sigma)}\left[
\mathcal{N}(\sigma)\right]  ^{\left(  1+it\right)  \gamma^{\prime}/2}%
U\sigma^{-\left(  1+it\right)  \gamma^{\prime}/2}W_{\sigma}\rho^{1/2}%
\right\Vert _{2\gamma}\\
&  \leq\max_{V_{\mathcal{N}(\sigma)},V_{\sigma}}\left\Vert \left[
\mathcal{N}(\rho)\right]  ^{-\gamma^{\prime}/2}V_{\mathcal{N}(\sigma)}\left[
\mathcal{N}(\sigma)\right]  ^{\gamma^{\prime}/2}U\sigma^{-\gamma^{\prime}%
/2}V_{\sigma}\rho^{1/2}\right\Vert _{2\gamma}\\
&  =\exp\left\{  \frac{\gamma^{\prime}}{2}\widetilde{\Delta}_{\gamma}^{\prime
}(\rho,\sigma,\mathcal{N})\right\}  ,\\
\left\Vert G\left(  \theta\right)  \right\Vert _{2\alpha}  &  =\left\Vert
\left[  \mathcal{N}(\rho)\right]  ^{-\alpha^{\prime}/2}W_{\mathcal{N}(\sigma
)}\left[  \mathcal{N}(\sigma)\right]  ^{\alpha^{\prime}/2}U\sigma
^{-\alpha^{\prime}/2}W_{\sigma}\rho^{1/2}\right\Vert _{2\alpha}.
\end{align}
Applying Theorem~\ref{thm:hadamard}, we find that the following inequality
holds for all $W_{\mathcal{N}(\sigma)}\in\mathbb{V}_{\mathcal{N}(\sigma)}$ and
$W_{\sigma}\in\mathbb{V}_{\sigma}$:%
\begin{equation}
\left\Vert \left[  \mathcal{N}(\rho)\right]  ^{-\alpha^{\prime}/2}%
W_{\mathcal{N}(\sigma)}\left[  \mathcal{N}(\sigma)\right]  ^{\alpha^{\prime
}/2}U\sigma^{-\alpha^{\prime}/2}W_{\sigma}\rho^{1/2}\right\Vert _{2\alpha}%
\leq\left[  \exp\left\{  \frac{\gamma^{\prime}}{2}\widetilde{\Delta}_{\gamma
}^{\prime}(\rho,\sigma,\mathcal{N})\right\}  \right]  ^{\frac{\alpha^{\prime}%
}{\gamma^{\prime}}}. \label{eq:other-chain-last}%
\end{equation}
We can then take a maximum over all $W_{\mathcal{N}(\sigma)}\in\mathbb{V}%
_{\mathcal{N}(\sigma)}$ and $W_{\sigma}\in\mathbb{V}_{\sigma}$ and apply the
definition in (\ref{eq:Delta-tilde-new}) to establish that%
\begin{equation}
\exp\left\{  \frac{\alpha^{\prime}}{2}\widetilde{\Delta}_{\alpha}^{\prime
}(\rho,\sigma,\mathcal{N})\right\}  \leq\left[  \exp\left\{  \frac
{\gamma^{\prime}}{2}\widetilde{\Delta}_{\gamma}^{\prime}(\rho,\sigma
,\mathcal{N})\right\}  \right]  ^{\frac{\alpha^{\prime}}{\gamma^{\prime}}}.
\end{equation}
The inequality in (\ref{eq:Delta-tilde-monotone}) then follows for
$1<\alpha<\gamma$ after taking a logarithm.

To get the monotonicity for the range $1/2\leq\alpha<\gamma<1$, we exchange
$\alpha$ and $\gamma$ in (\ref{eq:other-G})-(\ref{eq:other-theta}) and apply
the same reasoning as in (\ref{eq:other-chain-1})-(\ref{eq:other-chain-last})
to arrive at the following inequality:%
\begin{equation}
\exp\left\{  \frac{\gamma^{\prime}}{2}\widetilde{\Delta}_{\gamma}^{\prime
}(\rho,\sigma,\mathcal{N})\right\}  \leq\left[  \exp\left\{  \frac
{\alpha^{\prime}}{2}\widetilde{\Delta}_{\alpha}^{\prime}(\rho,\sigma
,\mathcal{N})\right\}  \right]  ^{\frac{\gamma\prime}{\alpha^{\prime}}}.
\end{equation}
Taking a negative logarithm and noting that $1/2\leq\alpha<\gamma<1$, so that
$\alpha^{\prime},\gamma^{\prime}\in\lbrack-1,0)$, then gives
(\ref{eq:Delta-tilde-monotone}) for this range.

We are now left with proving the case $\alpha\in\lbrack1/2,1)$ and $\gamma
\in(1,\infty]$ the dual parameter of $\alpha$: such that $1/\alpha+1/\gamma
=2$. Notice that $\alpha^{\prime}=-\gamma^{\prime}$ and we have that
$\gamma^{\prime}>0$. We pick%
\begin{align}
G\left(  z\right)   &  =\left[  \mathcal{N}(\rho)\right]  ^{-\left(
1-2z\right)  \alpha^{\prime}/2}\left[  \mathcal{N}(\sigma)\right]  ^{\left(
1-2z\right)  \alpha^{\prime}/2}U\sigma^{-\left(  1-2z\right)  \alpha^{\prime
}/2}\rho^{1/2},\\
p_{0}  &  =2\alpha,\\
p_{1}  &  =2\gamma,\\
\theta &  =1/2,
\end{align}
so that $p_{\theta}=2$. Consider that%
\begin{align}
\left\Vert G\left(  \theta\right)  \right\Vert _{2}  &  =\left\Vert \left[
\mathcal{N}(\rho)\right]  ^{-\left(  1-2\theta\right)  \alpha^{\prime}%
/2}\left[  \mathcal{N}(\sigma)\right]  ^{\left(  1-2\theta\right)
\alpha^{\prime}/2}U\sigma^{-\left(  1-2\theta\right)  \alpha^{\prime}/2}%
\rho^{1/2}\right\Vert _{2}\\
&  =\left\Vert U\rho^{1/2}\right\Vert _{2}=\left\Vert \rho^{1/2}\right\Vert
_{2}=1.
\end{align}
We then find that%
\begin{align}
M_{0}  &  =\sup_{t\in\mathbb{R}}\left\Vert G\left(  it\right)  \right\Vert
_{2\alpha}\\
&  =\sup_{t\in\mathbb{R}}\left\Vert \left[  \mathcal{N}(\rho)\right]
^{-\left(  1-2it\right)  \alpha^{\prime}/2}\left[  \mathcal{N}(\sigma)\right]
^{\left(  1-2it\right)  \alpha^{\prime}/2}U\sigma^{-\left(  1-2it\right)
\alpha^{\prime}/2}\rho^{1/2}\right\Vert _{2\alpha}\\
&  \leq\max_{V_{\mathcal{N}(\sigma)},V_{\sigma}}\left\Vert \left[
\mathcal{N}(\rho)\right]  ^{-\alpha^{\prime}/2}V_{\mathcal{N}(\sigma)}\left[
\mathcal{N}(\sigma)\right]  ^{\alpha^{\prime}/2}U\sigma^{-\alpha^{\prime}%
/2}V_{\sigma}\rho^{1/2}\right\Vert _{2\alpha}\\
&  =\exp\left\{  \frac{\alpha^{\prime}}{2}\widetilde{\Delta}_{\alpha}^{\prime
}(\rho,\sigma,\mathcal{N})\right\}  .
\end{align}
Consider that%
\begin{equation}
\left(  1-2\left(  1+it\right)  \right)  \alpha^{\prime}=-\left(
1+2it\right)  \alpha^{\prime}=\left(  1+2it\right)  \gamma^{\prime}.
\end{equation}
Thus, similarly, we have%
\begin{align}
M_{1}  &  =\sup_{t\in\mathbb{R}}\left\Vert G\left(  1+it\right)  \right\Vert
_{2\gamma}\\
&  =\sup_{t\in\mathbb{R}}\left\Vert \left[  \mathcal{N}(\rho)\right]
^{-\left(  1+2it\right)  \gamma^{\prime}/2}\left[  \mathcal{N}(\sigma)\right]
^{\left(  1+2it\right)  \gamma^{\prime}/2}U\sigma^{-\left(  1+2it\right)
\gamma^{\prime}/2}\rho^{1/2}\right\Vert _{2\gamma}\\
&  \leq\max_{V_{\mathcal{N}(\sigma)},V_{\sigma}}\left\Vert \left[
\mathcal{N}(\rho)\right]  ^{-\gamma^{\prime}/2}V_{\mathcal{N}(\sigma)}\left[
\mathcal{N}(\sigma)\right]  ^{\gamma^{\prime}/2}U\sigma^{-\gamma^{\prime}%
/2}V_{\sigma}\rho^{1/2}\right\Vert _{2\gamma}\\
&  =\exp\left\{  \frac{\gamma^{\prime}}{2}\widetilde{\Delta}_{\gamma}^{\prime
}(\rho,\sigma,\mathcal{N})\right\}  .
\end{align}
Applying Theorem~\ref{thm:hadamard} gives%
\begin{align}
1  &  \leq\exp\left\{  \frac{\alpha^{\prime}}{4}\widetilde{\Delta}_{\alpha
}^{\prime}(\rho,\sigma,\mathcal{N})\right\}  \exp\left\{  \frac{\gamma
^{\prime}}{4}\widetilde{\Delta}_{\gamma}^{\prime}(\rho,\sigma,\mathcal{N}%
)\right\} \\
&  =\exp\left\{  -\frac{\gamma^{\prime}}{4}\widetilde{\Delta}_{\alpha}%
^{\prime}(\rho,\sigma,\mathcal{N})\right\}  \exp\left\{  \frac{\gamma^{\prime
}}{4}\widetilde{\Delta}_{\gamma}^{\prime}(\rho,\sigma,\mathcal{N})\right\}  ,
\end{align}
which implies (\ref{eq:Delta-tilde-monotone}) for $\alpha\in\lbrack1/2,1)$ and
$1/\gamma=2-1/\alpha$. Putting the three cases together along with
Proposition~\ref{prop:lim-a-1}\ gives the inequality in
(\ref{eq:Delta-tilde-monotone}) for $1/2\leq\alpha\leq\gamma\leq\infty$,
$\alpha\neq1$, and $\gamma\neq1$.
\end{proof}

\subsection{Bounds for the quantum relative entropy difference}

A recent work \cite{W15}\ established refinements of the monotonicity of
quantum relative entropy, strong subadditivity, and other entropy
inequalities. In this section, we point out that these results follow as a
consequence of the properties of the swiveled R\'{e}nyi entropies and along
the way establish two new refinements of these entropy inequalities.

We begin with a brief background. Let $\mathcal{P}_{\sigma,\mathcal{N}}$
denote the Petz recovery map \cite{Petz1986,Petz1988} (see also~\cite{BK02}):%
\begin{equation}
\mathcal{P}_{\sigma,\mathcal{N}}(\cdot)\equiv\sigma^{1/2}\mathcal{N}^{\dag
}\left(  \left[  \mathcal{N}(\sigma)\right]  ^{-1/2}(\cdot)\left[
\mathcal{N}(\sigma)\right]  ^{-1/2}\right)  \sigma^{1/2}%
,\label{eq:Petz-channel-Rel-ent}%
\end{equation}
and let $\mathcal{R}_{\sigma,\mathcal{N}}^{V,W}$ denote the swiveled Petz
recovery map%
\begin{equation}
\mathcal{R}_{\sigma,\mathcal{N}}^{V,W}(\cdot)\equiv\left(  \mathcal{W}%
_{\sigma}\circ\mathcal{P}_{\sigma,\mathcal{N}}\circ\mathcal{V}_{\mathcal{N}%
(\sigma)}\right)  (\cdot),
\end{equation}
where the partial isometric map $\mathcal{V}_{\mathcal{N}(\sigma)}$ is defined
by%
\begin{equation}
\mathcal{V}_{\mathcal{N}(\sigma)}(\cdot)=V_{\mathcal{N}(\sigma)}%
(\cdot)V_{\mathcal{N}(\sigma)}^{\dag},\label{eq:unitaries}%
\end{equation}
and similarly for $\mathcal{W}_{\sigma}$, so that $\mathcal{V}_{\mathcal{N}%
(\sigma)}\left(  \mathcal{N}(\sigma)\right)  =\mathcal{N}(\sigma)$ and
$\mathcal{W}_{\sigma}(\sigma)=\sigma$. Observe then that%
\begin{equation}
\mathcal{R}_{\sigma,\mathcal{N}}^{V,W}\left(  \mathcal{N}(\sigma)\right)
=\sigma.
\end{equation}

Consider that particular values of $\alpha$ for $\Delta_{\alpha}^{\prime}%
(\rho,\sigma,\mathcal{N})$ and $\widetilde{\Delta}_{\alpha}^{\prime}%
(\rho,\sigma,\mathcal{N})$ lead to the following quantities, which can be
interpreted as a (pseudo-)distance from the state $\rho$ to the state
$\mathcal{N}(\rho)$ after a recovery channel $\mathcal{R}_{\sigma,\mathcal{N}%
}^{V,W}$ is applied:%
\begin{align}
\Delta_{0}^{\prime}(\rho,\sigma,\mathcal{N}) &  =\min_{V_{\mathcal{N}(\sigma
)},W_{\sigma}}D_{0}\left(  \rho\middle\Vert\mathcal{R}_{\sigma,\mathcal{N}%
}^{V,W}\left(  \mathcal{N}(\rho)\right)  \right)  ,\\
\widetilde{\Delta}_{1/2}^{\prime}(\rho,\sigma,\mathcal{N}) &  =-\log
\max_{V_{\mathcal{N}(\sigma)},W_{\sigma}}F\left(  \rho,\mathcal{R}%
_{\sigma,\mathcal{N}}^{V,W}\left(  \mathcal{N}\left(  \rho\right)  \right)
\right)  .\label{eq:fid-recover}%
\end{align}
These observations combined with the monotonicity from Theorems
\ref{thm:monotone}\ and \ref{thm:monotone-tilde}\ and the facts that
$D_{0}(\rho\Vert\mathcal{R}_{\sigma,\mathcal{N}}^{V,W}\left(  \mathcal{N}%
(\rho)\right)  )\geq0$ and $-\log\max_{V_{\mathcal{N}(\sigma)},W_{\sigma}%
}F(\rho,\mathcal{R}_{\sigma,\mathcal{N}}^{V,W}\left(  \mathcal{N}\left(
\rho\right)  \right)  )\geq0$\ allow us to conclude the following:

\begin{corollary}
Let $\rho$, $\sigma$, and $\mathcal{N}$ be as given in
Definition~\ref{def:rho-sig-N}. The swiveled R\'{e}nyi quantity $\Delta
_{\alpha}^{\prime}(\rho,\sigma,\mathcal{N})$ is non-negative for $\alpha
\in\lbrack0,1)\cup(1,2]$ and $\widetilde{\Delta}_{\alpha}^{\prime}(\rho
,\sigma,\mathcal{N})$ is non-negative for $\alpha\in\lbrack1/2,1)\cup
(1,\infty]$.
\end{corollary}

In order to establish the upper bounds in this section, we need to take $\rho
$, $\sigma$, and $\mathcal{N}$ as given in the following definition:

\begin{definition}
\label{def:rho-sig-N-2}Let $\rho_{SE^{\prime}}$ be a positive definite density
operator and let $\sigma_{SE^{\prime}}$ be a positive definite operator, each
acting on a finite-dimensional tensor-product Hilbert space $\mathcal{H}%
_{S}\otimes\mathcal{H}_{E^{\prime}}$. Let $\mathcal{N}$ be a quantum channel
given as follows:%
\begin{equation}
\mathcal{N}\left(  \theta_{SE^{\prime}}\right)  =\operatorname{Tr}_{E}\left\{
U_{SE^{\prime}\rightarrow BE}\theta_{SE^{\prime}}U_{SE^{\prime}\rightarrow
BE}^{\dag}\right\}  ,
\end{equation}
where $U_{SE^{\prime}\rightarrow BE}$ is a unitary operator taking
$\mathcal{H}_{S}\otimes\mathcal{H}_{E^{\prime}}$ to an isomorphic
finite-dimensional tensor-product Hilbert space $\mathcal{H}_{B}%
\otimes\mathcal{H}_{E}$, such that $\mathcal{N}( \rho) $ and $\mathcal{N}(
\sigma) $ are each positive definite and act on~$\mathcal{H}_{B}$.
\end{definition}

If $\rho$, $\sigma$, and $\mathcal{N}$ are taken as in
Definition~\ref{def:rho-sig-N-2}, then the following relations hold%
\begin{align}
\Delta_{2}^{\prime}( \rho,\sigma,\mathcal{N})  &  =\max_{V_{\mathcal{N}(
\sigma) },W_{\sigma}}D_{2}\left(  \rho\middle\Vert\mathcal{R}_{\sigma
,\mathcal{N}}^{V,W}\left(  \mathcal{N}( \rho) \right)  \right)  ,\\
\widetilde{\Delta}_{\infty}^{\prime}( \rho,\sigma,\mathcal{N})  &
=\max_{V_{\mathcal{N}( \sigma) },W_{\sigma}}D_{\max}\left(  \rho
\middle\Vert\mathcal{R}_{\sigma,\mathcal{N}}^{V,W}\left(  \mathcal{N}( \rho)
\right)  \right)  . \label{eq:dmax-recover}%
\end{align}

The main contribution of the recent work \cite{W15}\ was to show that the
relative entropy difference $\Delta(\rho,\sigma,\mathcal{N})$ in
(\ref{eq:rel-ent-diff}) can be bounded from below by (\ref{eq:fid-recover}).
In the case that $\rho$, $\sigma$, and $\mathcal{N}$ are taken as in
Definition~\ref{def:rho-sig-N-2}, then $\Delta(\rho,\sigma,\mathcal{N})$ can
be bounded from above by (\ref{eq:dmax-recover}). We find here that these
results are an immediate corollary of Proposition~\ref{prop:lim-a-1}\ and
Theorem~\ref{thm:monotone},\ and we also obtain two new bounds on $\Delta
(\rho,\sigma,\mathcal{N})$\ in terms of $\Delta_{0}(\rho,\sigma,\mathcal{N})$
and $\Delta_{2}(\rho,\sigma,\mathcal{N})$:

\begin{corollary}
\label{cor:recover-statement}Let $\rho$, $\sigma$, and $\mathcal{N}$ be as
given in Definition~\ref{def:rho-sig-N} and such that $\operatorname{supp}%
(\rho)\subseteq\operatorname{supp}(\sigma)$. Then the following inequalities
hold%
\begin{align}
-\log\max_{V_{\mathcal{N}(\sigma)},W_{\sigma}}F\left(  \rho,\mathcal{R}%
_{\sigma,\mathcal{N}}^{V,W}\left(  \mathcal{N}\left(  \rho\right)  \right)
\right)   &  \leq D(\rho\Vert\sigma)-D\left(  \mathcal{N}(\rho)\Vert
\mathcal{N}(\sigma)\right)  ,\\
\min_{V_{\mathcal{N}(\sigma)},W_{\sigma}}D_{0}\left(  \rho\middle\Vert
\mathcal{R}_{\sigma,\mathcal{N}}^{V,W}\left(  \mathcal{N}(\rho)\right)
\right)   &  \leq D(\rho\Vert\sigma)-D\left(  \mathcal{N}(\rho)\Vert
\mathcal{N}(\sigma)\right)  . \label{eq:D_0_V_bnd}%
\end{align}
If $\rho$, $\sigma$, and $\mathcal{N}$ are as given in
Definition~\ref{def:rho-sig-N-2}, then the following inequalities hold%
\begin{align}
D(\rho\Vert\sigma)-D\left(  \mathcal{N}(\rho)\Vert\mathcal{N}(\sigma)\right)
&  \leq\max_{V_{\mathcal{N}(\sigma)},W_{\sigma}}D_{\max}\left(  \rho
\middle\Vert\mathcal{R}_{\sigma,\mathcal{N}}^{V,W}\left(  \mathcal{N}%
(\rho)\right)  \right)  ,\\
D(\rho\Vert\sigma)-D\left(  \mathcal{N}(\rho)\Vert\mathcal{N}(\sigma)\right)
&  \leq\max_{V_{\mathcal{N}(\sigma)},W_{\sigma}}D_{2}\left(  \rho
\middle\Vert\mathcal{R}_{\sigma,\mathcal{N}}^{V,W}\left(  \mathcal{N}%
(\rho)\right)  \right)  .
\end{align}

\end{corollary}

As discussed in \cite{W15} (see also \cite{BLW14,SBW14}),
Corollary~\ref{cor:recover-statement}\ can be viewed as providing a physically
meaningful refinement of the monotonicity of quantum relative entropy in
(\ref{eq:mon-rel-ent}). The bound%
\begin{equation}
-\log\max_{V_{\mathcal{N}(\sigma)},W_{\sigma}}F\left(  \rho,\mathcal{R}%
_{\sigma,\mathcal{N}}^{V,W}\left(  \mathcal{N}\left(  \rho\right)  \right)
\right)  \leq D(\rho\Vert\sigma)-D\left(  \mathcal{N}(\rho)\Vert
\mathcal{N}(\sigma)\right)
\end{equation}
shows that if the decrease in relative entropy is small after the channel
$\mathcal{N}$ acts, then it is possible to perform the recovery map
$\mathcal{R}_{\sigma,\mathcal{N}}^{V,W}$ such that $\sigma$ is recovered
perfectly from $\mathcal{N}(\sigma)$, while the recovery of $\rho$ from
$\mathcal{N}(\rho)$ has a performance limited by the bound above. This result
has far reaching implications in quantum information theory as discussed in
\cite{W15} (see also \cite{BLW14,SBW14}).

We mention here that it is also possible to obtain bounds of the form from
\cite{W15}, with a single \textquotedblleft time\textquotedblright\ variable
$t\in\mathbb{R}$. The method of proof is similar to that for Theorem~4 in
\cite{W15}, so we give it in Appendix~\ref{app:auxiliary}. The formal
statement is as follows:

\begin{theorem}
\label{thm:rel-ent-other}Let $\rho$, $\sigma$, and $\mathcal{N}$ be as given
in Definition~\ref{def:rho-sig-N} and such that $\operatorname{supp}%
(\rho)\subseteq\operatorname{supp}(\sigma)$. Then the following inequalities
hold%
\begin{equation}
\inf_{t\in\mathbb{R}}D_{0}\left(  \rho\middle\Vert\mathcal{R}_{\sigma
,\mathcal{N}}^{t}\left(  \mathcal{N(}\rho)\right)  \right)  \leq D(\rho
\Vert\sigma)-D\left(  \mathcal{N(}\rho)\Vert\mathcal{N(}\sigma)\right)
,\label{eq:rel-ent-ineq}%
\end{equation}
where $\mathcal{R}_{\sigma,\mathcal{N}}^{t}$ is the following rotated Petz
recovery map:%
\begin{equation}
\mathcal{R}_{\sigma,\mathcal{N}}^{t}(\cdot)\equiv\left(  \mathcal{U}%
_{\sigma,t}\circ\mathcal{P}_{\sigma,\mathcal{N}}\circ\mathcal{U}%
_{\mathcal{N}(\sigma),-t}\right)  \left(  \cdot\right)
,\label{eq:rotated-Petz}%
\end{equation}
$\mathcal{P}_{\sigma,\mathcal{N}}$ is the Petz recovery map defined in
\eqref{eq:Petz-channel-Rel-ent}, and $\mathcal{U}_{\sigma,t}$ and
$\mathcal{U}_{\mathcal{N}(\sigma),-t}$ are partial isometric maps defined from%
\begin{equation}
\mathcal{U}_{\omega,t}(\cdot)\equiv\omega^{it}\left(  \cdot\right)
\omega^{-it},
\end{equation}
with $\omega$ a positive semi-definite operator. If $\rho$, $\sigma$, and
$\mathcal{N}$ are as given in Definition~\ref{def:rho-sig-N-2}, then%
\begin{equation}
D(\rho\Vert\sigma)-D\left(  \mathcal{N(}\rho)\Vert\mathcal{N(}\sigma)\right)
\leq\sup_{t\in\mathbb{R}}D_{2}\left(  \rho\middle\Vert\mathcal{R}%
_{\sigma,\mathcal{N}}^{t}(\mathcal{N(}\rho))\right)  .
\end{equation}

\end{theorem}

\begin{remark}
Note that it is possible to establish ``universal'' versions of the above
inequalities, by employing Hirschman's improvement \cite{H52} of the Hadamard
three-line theorem, as done in \cite{JRSWW15}.
\end{remark}

\section{Swiveled R\'{e}nyi conditional mutual information}

\label{sec:Renyi-CMI}In this section, we show how swiveled R\'{e}nyi
conditional mutual informations are special cases of the quantities defined in
the previous section. Furthermore, they satisfy some of the properties that
one would expect to hold for a R\'{e}nyi generalization of the conditional
mutual information. However, they generally do not converge to the conditional
mutual information in the limit as $\alpha\rightarrow1$.

Let $\rho_{ABC}$ be a density operator. Following from the observation
\cite{LW14} that%
\begin{equation}
I( A;B|C) _{\rho}=\Delta( \rho,\sigma,\mathcal{N}) ,
\end{equation}
for the choices%
\begin{equation}
\rho=\rho_{ABC},\ \ \ \ \ \sigma=\rho_{AC}\otimes I_{B},\ \ \ \ \ \mathcal{N}%
=\operatorname{Tr}_{A},
\end{equation}
we define the R\'{e}nyi conditional mutual informations to be a special case
of $\Delta_{\alpha}^{\prime}( \rho,\sigma,\mathcal{N}) $ and $\widetilde
{\Delta}_{\alpha}^{\prime}( \rho,\sigma,\mathcal{N}) $. That is, by setting%
\begin{align}
I_{\alpha}^{\prime}( A;B|C) _{\rho}  &  =\Delta_{\alpha}^{\prime}( \rho
_{ABC},\rho_{AC}\otimes I_{B},\operatorname{Tr}_{A}) ,\label{eq:renyi-cmi-1}\\
\widetilde{I}_{\alpha}^{\prime}( A;B|C) _{\rho}  &  =\widetilde{\Delta
}_{\alpha}^{\prime}( \rho_{ABC},\rho_{AC}\otimes I_{B},\operatorname{Tr}_{A})
, \label{eq:renyi-cmi-2}%
\end{align}
we obtain the swiveled R\'{e}nyi conditional mutual informations stated in the
following definition:

\begin{definition}
\label{def:swiveled-CMI}The swiveled R\'{e}nyi conditional mutual informations
are defined for a density operator $\rho_{ABC}$ and $\alpha\in\left(
0,1\right)  \cup\left(  1,\infty\right)  $ as follows:%
\begin{align}
I_{\alpha}^{\prime}( A;B|C) _{\rho}  &  \equiv\frac{2}{\alpha-1}\max
_{V_{\rho_{AC}},V_{\rho_{C}}}\log\left\Vert \rho_{BC}^{\left(  1-\alpha
\right)  /2}V_{\rho_{C}}\rho_{C}^{\left(  \alpha-1\right)  /2}\rho
_{AC}^{\left(  1-\alpha\right)  /2}V_{\rho_{AC}}\rho_{ABC}^{\alpha
/2}\right\Vert _{2},\\
\widetilde{I}_{\alpha}^{\prime}( A;B|C) _{\rho}  &  \equiv\frac{2}%
{\alpha^{\prime}}\max_{V_{\rho_{AC}},V_{\rho_{C}}}\log\left\Vert \rho
_{BC}^{-\alpha^{\prime}/2}V_{\rho_{C}}\rho_{C}^{\alpha^{\prime}/2}\rho
_{AC}^{-\alpha^{\prime}/2}V_{\rho_{AC}}\rho_{ABC}^{1/2}\right\Vert _{2\alpha},
\end{align}
where $\alpha^{\prime}=\left(  \alpha-1\right)  /\alpha$.
\end{definition}

We can now easily show that the R\'{e}nyi conditional mutual informations as
defined above satisfy several natural properties, with the exception of
convergence to the von Neumann conditional mutual information.

The following is a consequence of (\ref{eq:renyi-cmi-1}%
)--(\ref{eq:renyi-cmi-2}) and Proposition~\ref{prop:lim-a-1}:

\begin{corollary}
Let $\rho_{ABC}$ be a positive definite density operator. Then%
\begin{equation}
\lim_{\alpha\nearrow1}I_{\alpha}^{\prime}( A;B|C) _{\rho}=\lim_{\alpha
\nearrow1}\widetilde{I}_{\alpha}^{\prime}( A;B|C) _{\rho}\leq I( A;B|C)
_{\rho}\leq\lim_{\alpha\searrow1}I_{\alpha}^{\prime}( A;B|C) _{\rho}%
=\lim_{\alpha\searrow1}\widetilde{I}_{\alpha}^{\prime}( A;B|C) _{\rho}.
\end{equation}

\end{corollary}

They are monotone non-decreasing with respect to the parameter $\alpha$, which
follows from (\ref{eq:renyi-cmi-1})--(\ref{eq:renyi-cmi-2}) and
Theorems~\ref{thm:monotone} and \ref{thm:monotone-tilde}:

\begin{corollary}
\label{cor:renyi-cmi-ordered}Let $\rho_{ABC}$ be a density operator. The
swiveled R\'{e}nyi conditional mutual informations $I_{\alpha}^{\prime}(
A;B|C) _{\rho}$ and $\widetilde{I}_{\alpha}^{\prime}( A;B|C) _{\rho}$ are
monotone non-decreasing with respect to the R\'{e}nyi parameter for particular
values. For $0\leq\alpha\leq\gamma\leq2$, $\alpha\neq1$, and $\gamma\neq1$, we
have that%
\begin{equation}
I_{\alpha}^{\prime}( A;B|C) _{\rho}\leq I_{\gamma}^{\prime}( A;B|C) _{\rho},
\end{equation}
and for $1/2\leq\alpha\leq\gamma\leq\infty$, $\alpha\neq1$, and $\gamma\neq1$,%
\begin{equation}
\widetilde{I}_{\alpha}^{\prime}( A;B|C) _{\rho}\leq\widetilde{I}_{\gamma
}^{\prime}( A;B|C) _{\rho}.
\end{equation}

\end{corollary}

They are monotone non-increasing with respect to a quantum channel acting on
the $B$ system, which follows by invoking \cite[Lemmas 13 and 23]{BSW14}:

\begin{corollary}
Let $\rho_{ABC}$ be a positive definite density operator, and let
$\mathcal{N}_{B\rightarrow B^{\prime}}$ be a quantum channel such that%
\begin{equation}
\sigma_{AB^{\prime}C}\equiv\mathcal{N}_{B\rightarrow B^{\prime}}( \rho_{ABC})
\end{equation}
is a positive definite density operator. Then for all $\alpha\in
\lbrack0,1)\cup(1,2]$, the following inequality holds%
\begin{equation}
I_{\alpha}^{\prime}( A;B|C) _{\rho}\geq I_{\alpha}^{\prime}( A;B^{\prime}|C)
_{\sigma},
\end{equation}
and for all $\alpha\in\lbrack1/2,1)\cup(1,\infty]$, the following inequality
holds%
\begin{equation}
\widetilde{I}_{\alpha}^{\prime}( A;B|C) _{\rho}\geq\widetilde{I}_{\alpha
}^{\prime}( A;B^{\prime}|C) _{\sigma}.
\end{equation}

\end{corollary}

Corollary~\ref{cor:renyi-cmi-ordered}, Proposition~\ref{prop:lim-a-1},\ and
(\ref{eq:renyi-cmi-1})--(\ref{eq:renyi-cmi-2}) then imply the following
refinements of the strong subaddivity of quantum entropy, two of which were
already determined in \cite{W15}:

\begin{corollary}
\label{thm:CMI}Let $\rho_{ABC}$ be a density operator. Then the following
inequalities hold%
\begin{align}
-\log\left[  \max_{W_{\rho_{C}},V_{\rho_{AC}}}F\left(  \rho_{ABC}%
,\mathcal{R}_{C\rightarrow AC}^{V,W}\left(  \rho_{BC}\right)  \right)
\right]   &  \leq I(A;B|C)_{\rho},\\
\min_{W_{\rho_{C}},V_{\rho_{AC}}}D_{0}\left(  \rho_{ABC}\middle\Vert
\mathcal{R}_{C\rightarrow AC}^{V,W}\left(  \rho_{BC}\right)  \right)   &  \leq
I(A;B|C)_{\rho},
\end{align}
where $\mathcal{R}_{C\rightarrow AC}^{V,W}$ is the following swiveled Petz
recovery map:%
\begin{equation}
\mathcal{R}_{C\rightarrow AC}^{V,W}(\cdot)\equiv\left(  \mathcal{V}_{\rho
_{AC}}\circ\mathcal{P}_{C\rightarrow AC}\circ\mathcal{W}_{\rho_{C}}\right)
(\cdot),\label{eq:CMI-petz-recovery}%
\end{equation}
the Petz recovery map $\mathcal{P}_{C\rightarrow AC}$ is defined as%
\begin{equation}
\mathcal{P}_{C\rightarrow AC}(\cdot)\equiv\mathcal{P}_{\rho_{AC}%
,\operatorname{Tr}_{A}}(\cdot)=\rho_{AC}^{1/2}\rho_{C}^{-1/2}(\cdot)\rho
_{C}^{-1/2}\rho_{AC}^{1/2},
\end{equation}
and the partial isometric maps $\mathcal{V}_{\rho_{AC}}$ and $\mathcal{W}%
_{\rho_{C}}$ are defined as in \eqref{eq:unitaries}. If $\rho_{ABC}$ is a
positive definite density operator, then the following inequalities hold%
\begin{align}
I(A;B|C)_{\rho} &  \leq\max_{W_{\rho_{C}},V_{\rho_{AC}}}D_{\max}\left(
\rho_{ABC}\middle\Vert\mathcal{R}_{C\rightarrow AC}^{V,W}\left(  \rho
_{BC}\right)  \right)  ,\\
I(A;B|C)_{\rho} &  \leq\max_{W_{\rho_{C}},V_{\rho_{AC}}}D_{2}\left(
\rho_{ABC}\middle\Vert\mathcal{R}_{C\rightarrow AC}^{V,W}\left(  \rho
_{BC}\right)  \right)  .
\end{align}

\end{corollary}

Note that remainder terms for strong subadditivity were put forward in
\cite{FR14,SOR15} before the recent developments in \cite{W15}.

\section{Swiveled R\'{e}nyi quantum information measures}

\label{sec:arbitrary-vn-measure}We now discuss how to extend the approach
given here and in \cite{BSW15a}\ in order to construct swiveled R\'{e}nyi
generalizations of any quantum information measure which consists of a linear
combination of von Neumann entropies with coefficients chosen from the set
$\left\{  -1,0,1\right\}  $. We repeat some of the discussions from
\cite{BSW15a}\ in order to illustrate the method.

Let $\rho_{A_{1}\cdots A_{l}}$ be a density operator on $l$ systems and set
$\mathcal{A}\equiv\left\{  A_{1},\ldots,A_{l}\right\}  $. Suppose that we
would like to establish a R\'{e}nyi generalization of the following linear
combination of entropies:%
\begin{equation}
L\left(  \rho_{A_{1}\cdots A_{l}}\right)  \equiv\sum_{S\in\mathcal{P}_{\geq
1}\left(  \mathcal{A}\right)  }a_{S}H(S)_{\rho}, \label{eq:linear-combo-vN}%
\end{equation}
where $\mathcal{P}_{\geq1}\left(  \mathcal{A}\right)  $ is the power set of
$\mathcal{A}$ (excluding the empty set), such that the sum runs over all
subsets of the systems $A_{1},\ldots,A_{l}$. Furthermore, each coefficient
$a_{S}\in\left\{  -1,0,1\right\}  $ and corresponds to a subset $S$. In the
case that $a_{\mathcal{A}}$ is nonzero, without loss of generality, we can set
$a_{\mathcal{A}}=-1$ (otherwise, factor out $-1$ to make this the case). Then,
we can rewrite the quantity in (\ref{eq:linear-combo-vN}) in terms of the
relative entropy as follows:%
\begin{equation}
D\left(  \rho_{A_{1}\cdots A_{l}}\middle\Vert\exp\left\{  \sum_{S\in
\mathcal{P}^{\prime}}a_{S}\log\rho_{S}\right\}  \right)  , \label{relentlin}%
\end{equation}
where $\mathcal{P}^{\prime}=\mathcal{P}_{\geq1}\left(  \mathcal{A}\right)
\backslash\left\{  A_{1},\ldots,A_{l}\right\}  $. On the other hand, if
$a_{\mathcal{A}}=0$, i.e., if all the marginal entropies in the sum are on a
number of systems that is strictly smaller than the number of systems over
which the state $\rho$ is defined (as is the case with $H(AB)+H(BC)+H(AC)$,
for example), we can take a purification of the original state and call this
purification the state $\rho_{A_{1}\cdots A_{l}}$. This state is now a pure
state on a number of systems strictly larger than the number of systems
involved in all the marginal entropies. We then add the entropy $H(A_{1}\ldots
A_{l})_{\rho}=0$ to the sum of entropies and apply the above recipe (so we
resolve the issue with this example by purifying to a system $R$, setting the
sum formula to be $H(ABCR)+H(AB)+H(BC)+H(AC)$, and proceeding with the above
recipe). In the case that the resulting density operator $\rho_{A_{1}\cdots
A_{l}}$\ is not positive definite, we can mix it with the maximally mixed
state $\pi_{A_{1}\cdots A_{l}}$\ as follows:%
\begin{equation}
\left(  1-\varepsilon\right)  \rho_{A_{1}\cdots A_{l}}+\varepsilon\pi
_{A_{1}\cdots A_{l}},
\end{equation}
where $\varepsilon\in(0,1)$. The resulting density operator is then
$\varepsilon$-distinguishable from the original one by any quantum measurement
performed on the systems $A_{1}\cdots A_{l}$.

We then define the following swiveled R\'{e}nyi entropies, which generalize
$L\left(  \rho_{A_{1}\cdots A_{l}}\right)  $ from \eqref{eq:linear-combo-vN}:%
\begin{align}
L_{\alpha}^{\prime}\left(  \rho_{A_{1}\cdots A_{l}}\right)   &  \equiv\frac
{2}{\alpha-1}\max_{\left\{  V_{\rho_{S}}\right\}  _{S}}\log\left\Vert \left[
\prod\limits_{S\in\mathcal{P}^{\prime}}\rho_{S}^{-a_{S}\left(  \alpha
-1\right)  /2}V_{\rho_{S}}\right]  \rho_{A_{1}\cdots A_{l}}^{\alpha
/2}\right\Vert _{2},\\
\widetilde{L}_{\alpha}^{\prime}\left(  \rho_{A_{1}\cdots A_{l}}\right)   &
\equiv\frac{2}{\alpha^{\prime}}\max_{\left\{  V_{\rho_{S}}\right\}  _{S}}%
\log\left\Vert \left[  \prod\limits_{S\in\mathcal{P}^{\prime}}\rho_{S}%
^{-a_{S}\alpha^{\prime}/2}V_{\rho_{S}}\right]  \rho_{A_{1}\cdots A_{l}}%
^{1/2}\right\Vert _{2\alpha},
\end{align}
where $\alpha^{\prime}=\left(  \alpha-1\right)  /\alpha$. The ordering of the
marginal density operators in the product is in general arbitrary, but could
be important for some applications (consider, e.g., that the choices in
Definition~\ref{def:swiveled-CMI}\ lead to the inequalities in
Corollary~\ref{thm:CMI}, which have a physical interpretation in terms of recovery).

By the same methods as given throughout this paper, we can establish several
properties of these quantities. It follows from the generalized Lie-Trotter
product formula \cite{S85} and the method given in the proof of
Proposition~\ref{prop:lim-a-1}\ that%
\begin{equation}
\lim_{\alpha\nearrow1}L_{\alpha}^{\prime}\left(  \rho_{A_{1}\cdots A_{l}%
}\right)  =\lim_{\alpha\nearrow1}\widetilde{L}_{\alpha}^{\prime}\left(
\rho_{A_{1}\cdots A_{l}}\right)  \leq L\left(  \rho_{A_{1}\cdots A_{l}%
}\right)  \leq\lim_{\alpha\searrow1}L_{\alpha}^{\prime}\left(  \rho
_{A_{1}\cdots A_{l}}\right)  =\lim_{\alpha\searrow1}\widetilde{L}_{\alpha
}^{\prime}\left(  \rho_{A_{1}\cdots A_{l}}\right)  .
\end{equation}
From the same method as given in the proof of Theorems~\ref{thm:monotone} and
\ref{thm:monotone-tilde}, for $0\leq\alpha\leq\gamma\leq2$, $\alpha\neq1$, and
$\gamma\neq1$, we can conclude that%
\begin{equation}
L_{\alpha}^{\prime}\left(  \rho_{A_{1}\cdots A_{l}}\right)  \leq L_{\gamma
}^{\prime}\left(  \rho_{A_{1}\cdots A_{l}}\right)  ,
\end{equation}
and for $1/2\leq\alpha\leq\gamma\leq\infty$, $\alpha\neq1$, and $\gamma\neq1$,%
\begin{equation}
\widetilde{L}_{\alpha}^{\prime}\left(  \rho_{A_{1}\cdots A_{l}}\right)
\leq\widetilde{L}_{\gamma}^{\prime}\left(  \rho_{A_{1}\cdots A_{l}}\right)  .
\end{equation}
The inequalities above then lead to the following bounds for $L\left(
\rho_{A_{1}\cdots A_{l}}\right)  $:%
\begin{align}
L_{0}^{\prime}\left(  \rho_{A_{1}\cdots A_{l}}\right)   &  \leq L\left(
\rho_{A_{1}\cdots A_{l}}\right)  \leq L_{2}^{\prime}\left(  \rho_{A_{1}\cdots
A_{l}}\right)  ,\\
\widetilde{L}_{1/2}^{\prime}\left(  \rho_{A_{1}\cdots A_{l}}\right)   &  \leq
L\left(  \rho_{A_{1}\cdots A_{l}}\right)  \leq\widetilde{L}_{\infty}^{\prime
}\left(  \rho_{A_{1}\cdots A_{l}}\right)  ,
\end{align}
which in some cases might have physical interpretations in terms of swiveled
Petz recovery channels (see \cite[Sections~5.6 and 5.7]{W15} for some examples).

\section{Monotonicity of Trace Quantities}

\cite{Z14b} posed an open question regarding whether the following quantity%
\begin{equation}
\text{Tr}\left\{  \left[  \rho_{AC}^{\left(  1-\alpha\right)  /2}\rho
_{C}^{\left(  \alpha-1\right)  /2}\rho_{BC}^{1-\alpha}\rho_{C}^{\left(
\alpha-1\right)  /2}\rho_{AC}^{\left(  1-\alpha\right)  /2}\right]
^{1/\left(  1-\alpha\right)  }\right\}  \label{eq:r-ABC-trace-quantity}%
\end{equation}
is monotone in $\alpha$ and never exceeds one. The recent work \cite{DW15}%
\ addressed this open question, first by generalizing it and then proving that%
\begin{equation}
\text{Tr}\left\{  \left[  \sigma^{\left(  1-\alpha\right)  /2}\mathcal{N}%
^{\dag}\left(  \mathcal{N}(\sigma)^{\left(  \alpha-1\right)  /2}%
\mathcal{N}(\rho)^{1-\alpha}\mathcal{N}(\sigma)^{\left(  \alpha-1\right)
/2}\right)  \sigma^{\left(  1-\alpha\right)  /2}\right]  ^{1/\left(
1-\alpha\right)  }\right\}  \leq1, \label{eq:r-s-N-trace-quantity}%
\end{equation}
for $\alpha\in(0,1)$ and $\rho$, $\sigma$, and $\mathcal{N}$ as given in
Definition~\ref{def:rho-sig-N}. The same work established that this bound
holds for $\alpha\in(1,2)$ if $\rho$, $\sigma$, and $\mathcal{N}$ are as given
in Definition~\ref{def:rho-sig-N}. One recovers the quantity in
(\ref{eq:r-ABC-trace-quantity}) by picking $\rho=\rho_{ABC}$, $\sigma
=\rho_{AC}\otimes I_{B}$, and $\mathcal{N} = \operatorname{Tr}_{A}$ in
(\ref{eq:r-s-N-trace-quantity}). It is not known whether the quantity in
(\ref{eq:r-s-N-trace-quantity}) is monotone with respect to $\alpha$.

Here, we address this latter question by again taking our approach of allowing
for a unitary swivel. Consider that we can rewrite the left-hand side of
(\ref{eq:r-s-N-trace-quantity}) as follows for $\alpha\in(0,1)$:%
\begin{equation}
\left\Vert \left[  \mathcal{N}(\rho)^{\left(  1-\alpha\right)  /2}%
\mathcal{N}(\sigma)^{\left(  \alpha-1\right)  /2}\otimes I_{E}\right]
U\sigma^{\left(  1-\alpha\right)  /2}\right\Vert _{2/\left(  1-\alpha\right)
}^{2/\left(  1-\alpha\right)  },
\end{equation}
where $U$ is an isometric extension of the channel $\mathcal{N}$. So we
instead consider the following quantity, which has an optimization over a
unitary swivel:%
\begin{equation}
\max_{V_{\mathcal{N}(\sigma)}}\left\Vert \left[  \mathcal{N}(\rho)^{\left(
1-\alpha\right)  /2}V_{\mathcal{N}(\sigma)}\mathcal{N}(\sigma)^{\left(
\alpha-1\right)  /2}\otimes I_{E}\right]  U\sigma^{\left(  1-\alpha\right)
/2}\right\Vert _{2/\left(  1-\alpha\right)  }^{2/\left(  1-\alpha\right)  }.
\end{equation}
To simplify the notation, consider that the above quantity for $\alpha
\in\lbrack0,1)$ is the same as the following one for $p\in\lbrack2,\infty)$:%
\begin{equation}
\max_{V_{\mathcal{N}(\sigma)}}\left\Vert \left[  \mathcal{N}(\rho
)^{1/p}V_{\mathcal{N}(\sigma)}\mathcal{N}(\sigma)^{-1/p}\otimes I_{E}\right]
U\sigma^{1/p}\right\Vert _{p}^{p}. \label{eq:r-s-N-trace-quant-p}%
\end{equation}
We can now state our contribution to the open question:

\begin{proposition}
\label{prop:zhang-conj}The quantity in \eqref{eq:r-s-N-trace-quant-p} is
monotone non-increasing on the interval $p\in\lbrack2,\infty)$ and has a
maximum value of one at $p=2$ if $\operatorname{supp}(\rho)\subseteq
\operatorname{supp}(\sigma)$.
\end{proposition}

\begin{proof}
This ends up being another application of the Hadamard three-line theorem,
using techniques similar to what we have used previously. For $q\in
\lbrack2,\infty)$, $q<p$, and $V_{\mathcal{N}(\sigma)}$ a fixed unitary
commuting with $\mathcal{N}(\sigma)$, pick%
\begin{align}
G\left(  z\right)   &  =\left[  \mathcal{N}(\rho)^{z/q}V_{\mathcal{N}(\sigma
)}\mathcal{N}(\sigma)^{-z/q}\otimes I_{E}\right]  U\sigma^{z/q},\\
p_{0}  &  =\infty,\\
p_{1}  &  =q,\\
\theta &  =q/p,
\end{align}
which implies that $p_{\theta}=p$. Applying Theorem~\ref{thm:hadamard}\ gives%
\begin{equation}
\left\Vert G\left(  \theta\right)  \right\Vert _{p}\leq\left[  \sup
_{t\in\mathbb{R}}\left\Vert G\left(  it\right)  \right\Vert _{\infty}\right]
^{1-\theta}\left[  \sup_{t\in\mathbb{R}}\left\Vert G\left(  1+it\right)
\right\Vert _{q}\right]  ^{\theta}. \label{eq:app-3-line-1}%
\end{equation}
So we evaluate these terms to find%
\begin{align}
\left\Vert G\left(  \theta\right)  \right\Vert _{p}  &  =\left\Vert \left[
\mathcal{N}(\rho)^{\theta/q}V_{\mathcal{N}(\sigma)}\mathcal{N}(\sigma
)^{-\theta/q}\otimes I_{E}\right]  U\sigma^{\theta/q}\right\Vert _{p}\\
&  =\left\Vert \left[  \mathcal{N}(\rho)^{1/p}V_{\mathcal{N}(\sigma
)}\mathcal{N}(\sigma)^{-1/p}\otimes I_{E}\right]  U\sigma^{1/p}\right\Vert
_{p},\\
\sup_{t\in\mathbb{R}}\left\Vert G\left(  it\right)  \right\Vert _{\infty}  &
=\sup_{t\in\mathbb{R}}\left\Vert \left[  \mathcal{N}(\rho)^{it/q}%
V_{\mathcal{N}(\sigma)}\mathcal{N}(\sigma)^{-it/q}\otimes I_{E}\right]
U\sigma^{it/q}\right\Vert _{\infty}\\
&  \leq1,\\
\sup_{t\in\mathbb{R}}\left\Vert G\left(  1+it\right)  \right\Vert _{q}  &
=\sup_{t\in\mathbb{R}}\left\Vert \left[  \mathcal{N}(\rho)^{\left(
1+it\right)  /q}V_{\mathcal{N}(\sigma)}\mathcal{N}(\sigma)^{-\left(
1+it\right)  /q}\otimes I_{E}\right]  U\sigma^{\left(  1+it\right)
/q}\right\Vert _{q}\\
&  =\sup_{t\in\mathbb{R}}\left\Vert \left[  \mathcal{N}(\rho)^{1/q}%
V_{\mathcal{N}(\sigma)}\mathcal{N}(\sigma)^{-it/q}\mathcal{N}(\sigma
)^{-1/q}\otimes I_{E}\right]  U\sigma^{1/q}\right\Vert _{q}\\
&  \leq\max_{W_{\mathcal{N}(\sigma)}}\left\Vert \left[  \mathcal{N}%
(\rho)^{1/q}W_{\mathcal{N}(\sigma)}\mathcal{N}(\sigma)^{-1/q}\otimes
I_{E}\right]  U\sigma^{1/q}\right\Vert _{q}.
\end{align}
Putting everything together, we find that for $2\leq q<p$, the following
inequality holds%
\begin{multline}
\max_{V_{\mathcal{N}(\sigma)}}\left\Vert \left[  \mathcal{N}(\rho
)^{1/p}V_{\mathcal{N}(\sigma)}\mathcal{N}(\sigma)^{-1/p}\otimes I_{E}\right]
U\sigma^{1/p}\right\Vert _{p}^{p}\\
\leq\max_{W_{\mathcal{N}(\sigma)}}\left\Vert \left[  \mathcal{N}(\rho
)^{1/q}W_{\mathcal{N}(\sigma)}\mathcal{N}(\sigma)^{-1/q}\otimes I_{E}\right]
U\sigma^{1/q}\right\Vert _{q}^{q},
\end{multline}
since the inequality from (\ref{eq:app-3-line-1}) holds for all
$V_{\mathcal{N}(\sigma)}$. This establishes the first statement in the proposition.

For the second statement, consider evaluating
\eqref{eq:r-s-N-trace-quant-p}\ at $p=2$ for any choice of $V_{\mathcal{N}%
(\sigma)}$:%
\begin{align}
&  \left\Vert \left[  \mathcal{N}(\rho)^{1/2}V_{\mathcal{N}(\sigma
)}\mathcal{N}(\sigma)^{-1/2}\otimes I_{E}\right]  U\sigma^{1/2}\right\Vert
_{2}^{2}\nonumber\\
&  =\operatorname{Tr}\left\{  \sigma^{1/2}U^{\dag}\left[  \mathcal{N}%
(\sigma)^{-1/2}V_{\mathcal{N}(\sigma)}^{\dag}\mathcal{N}(\rho)V_{\mathcal{N}%
(\sigma)}\mathcal{N}(\sigma)^{-1/2}\otimes I_{E}\right]  U\sigma^{1/2}\right\}
\\
&  =\operatorname{Tr}\left\{  \sigma^{1/2}\mathcal{N}^{\dag}\left[
\mathcal{N}(\sigma)^{-1/2}V_{\mathcal{N}(\sigma)}^{\dag}\mathcal{N}%
(\rho)V_{\mathcal{N}(\sigma)}\mathcal{N}(\sigma)^{-1/2}\right]  \sigma
^{1/2}\right\} \\
&  =\operatorname{Tr}\left\{  \mathcal{N}(\sigma)\mathcal{N}(\sigma
)^{-1/2}V_{\mathcal{N}(\sigma)}^{\dag}\mathcal{N}(\rho)V_{\mathcal{N}(\sigma
)}\mathcal{N}(\sigma)^{-1/2}\right\} \\
&  =\operatorname{Tr}\left\{  \Pi_{\mathcal{N}(\sigma)}V_{\mathcal{N}(\sigma
)}^{\dag}\mathcal{N}(\rho)V_{\mathcal{N}(\sigma)}\right\} \\
&  =\operatorname{Tr}\left\{  \Pi_{\mathcal{N}(\sigma)}\mathcal{N}%
(\rho)\right\} \\
&  =1.
\end{align}

\end{proof}

\begin{corollary}
Let $\rho_{ABC}$ be a density operator. Then the following quantity is
monotone non-increasing for $\alpha\in\lbrack0,1)$ and takes a maximum value
of one at $\alpha=0$:%
\begin{equation}
\max_{V_{\rho_{C}}}\operatorname{Tr}\left\{  \left(  \rho_{AC}^{\left(
1-\alpha\right)  /2}V_{\rho_{C}}\rho_{C}^{\left(  \alpha-1\right)  /2}%
\rho_{BC}^{1-\alpha}\rho_{C}^{\left(  \alpha-1\right)  /2}V_{\rho_{C}}^{\dag
}\rho_{AC}^{\left(  1-\alpha\right)  /2}\right)  ^{1/\left(  1-\alpha\right)
}\right\}  .
\end{equation}
If $\rho_{ABC}$ is a positive definite, then the same quantity is monotone
non-decreasing for $\alpha\in(1,2]$ and takes a maximum value of one at
$\alpha=2$.
\end{corollary}

\begin{proof}
The first statement follows by applying Proposition~\ref{prop:zhang-conj}\ for
the choices $\rho=\rho_{ABC}$, $\sigma=\rho_{AC}\otimes I_{B}$, and
$\mathcal{N}=\operatorname{Tr}_{A}$. The second statement follows because%
\begin{multline}
\left(  \rho_{AC}^{\left(  1-\alpha\right)  /2}V_{\rho_{C}}\rho_{C}^{\left(
\alpha-1\right)  /2}\rho_{BC}^{1-\alpha}\rho_{C}^{\left(  \alpha-1\right)
/2}V_{\rho_{C}}^{\dag}\rho_{AC}^{\left(  1-\alpha\right)  /2}\right)
^{1/\left(  1-\alpha\right)  }\\
=\left(  \rho_{AC}^{\left(  1-\beta\right)  /2}V_{\rho_{C}}\rho_{C}^{\left(
\beta-1\right)  /2}\rho_{BC}^{1-\beta}\rho_{C}^{\left(  \beta-1\right)
/2}V_{\rho_{C}}^{\dag}\rho_{AC}^{\left(  1-\beta\right)  /2}\right)
^{1/\left(  1-\beta\right)  },
\end{multline}
for $\beta=2-\alpha\in(0,1)$ and then we apply the first statement.
\end{proof}

\section{Conclusion}

\label{sec:conclusion}The main contribution of this paper is a general method,
the \textquotedblleft swiveled R\'{e}nyi entropic\textquotedblright\ approach,
for constructing $\alpha$-R\'{e}nyi generalizations of a quantum information
measure that are monotone non-decreasing in the parameter $\alpha$. The
swiveled R\'{e}nyi entropies are discontinuous at $\alpha=1$ and do not
converge to von Neumann entropy-based quantities in the limit as
$\alpha\rightarrow1$. We suspect that the swiveled R\'{e}nyi entropies might
be helpful in understanding refinements of quantum information-processing
tasks, but this remains unclear due to the lack of convergence at $\alpha=1$.
At the very least, the technique recovers the recent refinements
\cite{W15}\ of fundamental entropy inequalities such as monotonicity of
quantum relative entropy \cite{Lindblad1975,U77}\ and strong subadditivity
\cite{LR73,PhysRevLett.30.434}, in addition to providing new refinements for
these entropy inequalities.

The most important open question going forward from here is to determine
R\'{e}nyi entropies which satisfy all of the desiderata that one would have
for R\'{e}nyi generalizations of quantum information measures. We find it
curious that with the proposal in \cite{BSW15a}, one can prove convergence to
a von Neumann entropy-based quantity in the limit as $\alpha\rightarrow1$, but
we are still unable to establish monotonicity in $\alpha$. However, with the
swiveled R\'{e}nyi entropies proposed here, the situation is reversed.

One might also consider developing chain rules for the swiveled R\'{e}nyi
entropies, along the lines established in \cite{D14}, but it is unclear how
useful this would be in practice, given that the quantities do not generally
converge at $\alpha=1$.

\bigskip

\textbf{Acknowledgements.} We are grateful to Salman Beigi for insightful
discussions about the topic of this paper. FD acknowledges the support of the
Czech Science Foundation GA CR project P202/12/1142 and the support of the EU
FP7 under grant agreement no 323970 (RAQUEL). MMW is grateful to Stephanie
Wehner and her group for hospitality during a research visit to TU\ Delft (May
2015), to Renato Renner and his group for the same during a visit to ETH
Zurich (June 2015), and acknowledges support from startup funds from the
Department of Physics and Astronomy at LSU, the NSF\ under Award
No.~CCF-1350397, and the DARPA Quiness Program through US Army Research Office
award W31P4Q-12-1-0019.

\appendix

\section{Limit as $\alpha\rightarrow1$}

\label{app:Delta-limit-alpha-1}

\begin{definition}
Let $\rho$, $\sigma$, and $\mathcal{N}$ be as given in
Definition~\ref{def:rho-sig-N}. For $\alpha\in\left(  0,1\right)  \cup\left(
1,\infty\right)  $, let%
\begin{equation}
\Delta_{\alpha}(\rho,\sigma,\mathcal{N})=\frac{1}{\alpha-1}\log Q_{\alpha
}(\rho,\sigma,\mathcal{N}),
\end{equation}
where%
\begin{equation}
Q_{\alpha}(\rho,\sigma,\mathcal{N})\equiv\left\Vert \left(  \mathcal{N}%
(\rho)^{\left(  1-\alpha\right)  /2}\mathcal{N}(\sigma)^{\left(
\alpha-1\right)  /2}\otimes I_{E}\right)  U\sigma^{\left(  1-\alpha\right)
/2}\rho^{\alpha/2}\right\Vert _{2}^{2}.
\end{equation}

\end{definition}

\begin{theorem}
Let $\rho$, $\sigma$, and $\mathcal{N}$ be as given in
Definition~\ref{def:rho-sig-N} and such that $\operatorname{supp}%
(\rho)\subseteq\operatorname{supp}(\sigma)$. The following limit holds%
\begin{equation}
\lim_{\alpha\rightarrow1}\Delta_{\alpha}(\rho,\sigma,\mathcal{N})=D(\rho
\Vert\sigma)-D\left(  \mathcal{N}(\rho)\Vert\mathcal{N}(\sigma)\right)  .
\end{equation}

\end{theorem}

\begin{proof}
Let $\Pi_{\omega}$ denote the projection onto the support of $\omega$. From
the condition $\operatorname{supp}(\rho)\subseteq\operatorname{supp}(\sigma)$,
it follows that $\operatorname{supp}\left(  \mathcal{N}(\rho)\right)
\subseteq\operatorname{supp}\left(  \mathcal{N}(\sigma)\right)  $
\cite[Appendix~B.4]{R05}. We can then conclude that%
\begin{equation}
\Pi_{\sigma}\Pi_{\rho}=\Pi_{\rho},\ \ \ \ \ \ \ \ \Pi_{\mathcal{N}(\rho)}%
\Pi_{\mathcal{N}(\sigma)}=\Pi_{\mathcal{N}(\rho)}. \label{eq:Pi-rho}%
\end{equation}
We also know that $\operatorname{supp}\left(  U\rho U^{\dag}\right)
\subseteq\operatorname{supp}\left(  \mathcal{N}(\rho)\otimes I_{E}\right)  $
\cite[Appendix~B.4]{R05}, so that%
\begin{equation}
\left(  \Pi_{\mathcal{N}(\rho)}\otimes I_{E}\right)  \Pi_{U\rho U^{\dag}}%
=\Pi_{U\rho U^{\dag}}. \label{eq:Pi-U-rho}%
\end{equation}
When $\alpha=1$, we find from the above facts that%
\begin{align}
Q_{1}(\rho,\sigma,\mathcal{N})  &  =\left\Vert \left(  \Pi_{\mathcal{N}(\rho
)}\Pi_{\mathcal{N}(\sigma)}\otimes I_{E}\right)  U\Pi_{\sigma}\rho
^{1/2}\right\Vert _{2}^{2}\\
&  =\left\Vert \left(  \Pi_{\mathcal{N}(\rho)}\otimes I_{E}\right)  U\Pi
_{\rho}\rho^{1/2}\right\Vert _{2}^{2}\\
&  =\left\Vert \left(  \Pi_{\mathcal{N}(\rho)}\otimes I_{E}\right)  \Pi_{U\rho
U^{\dag}}U\rho^{1/2}\right\Vert _{2}^{2}\\
&  =\left\Vert \Pi_{U\rho U^{\dag}}U\rho^{1/2}\right\Vert _{2}^{2}\\
&  =\left\Vert \rho^{1/2}\right\Vert _{2}^{2}\\
&  =1.
\end{align}
So from the definition of the derivative, this means that%
\begin{align}
\lim_{\alpha\rightarrow1}\Delta_{\alpha}(\rho,\sigma,\mathcal{N})  &
=\lim_{\alpha\rightarrow1}\frac{\log Q_{\alpha}(\rho,\sigma,\mathcal{N})-\log
Q_{1}(\rho,\sigma,\mathcal{N})}{\alpha-1}\\
&  =\left.  \frac{d}{d\alpha}\left[  \log Q_{\alpha}(\rho,\sigma
,\mathcal{N})\right]  \right\vert _{\alpha=1}\\
&  =\frac{1}{Q_{1}(\rho,\sigma,\mathcal{N})}\left.  \frac{d}{d\alpha}\left[
Q_{\alpha}(\rho,\sigma,\mathcal{N})\right]  \right\vert _{\alpha=1}\\
&  =\left.  \frac{d}{d\alpha}\left[  Q_{\alpha}(\rho,\sigma,\mathcal{N}%
)\right]  \right\vert _{\alpha=1}. \label{eq:limit-a-1-exp}%
\end{align}
Let $\alpha^{\prime}\equiv\alpha-1$. Consider that%
\begin{equation}
Q_{\alpha}(\rho,\sigma,\mathcal{N})=\operatorname{Tr}\left\{  \rho^{\alpha
}\sigma^{-\alpha^{\prime}/2}\mathcal{N}^{\dag}\left(  \mathcal{N}%
(\sigma)^{\alpha^{\prime}/2}\mathcal{N}(\rho)^{-\alpha^{\prime}}%
\mathcal{N}(\sigma)^{\alpha^{\prime}/2}\right)  \sigma^{-\alpha^{\prime}%
/2}\right\}  .
\end{equation}
Now we calculate $\frac{d}{d\alpha}Q_{\alpha}(\rho,\sigma,\mathcal{N})$:%
\begin{multline}
\frac{d}{d\alpha}\operatorname{Tr}\left\{  \rho^{\alpha}\sigma^{-\alpha
^{\prime}/2}\mathcal{N}^{\dag}\left(  \mathcal{N}(\sigma)^{\alpha^{\prime}%
/2}\mathcal{N}(\rho)^{-\alpha^{\prime}}\mathcal{N}(\sigma)^{\alpha^{\prime}%
/2}\right)  \sigma^{-\alpha^{\prime}/2}\right\} \\
=\operatorname{Tr}\left\{  \left[  \frac{d}{d\alpha}\rho^{\alpha}\right]
\sigma^{-\alpha^{\prime}/2}\mathcal{N}^{\dag}\left(  \mathcal{N}%
(\sigma)^{\alpha^{\prime}/2}\mathcal{N}(\rho)^{-\alpha^{\prime}}%
\mathcal{N}(\sigma)^{\alpha^{\prime}/2}\right)  \sigma^{-\alpha^{\prime}%
/2}\right\} \\
+\operatorname{Tr}\left\{  \rho^{\alpha}\left[  \frac{d}{d\alpha}%
\sigma^{-\alpha^{\prime}/2}\right]  \mathcal{N}^{\dag}\left(  \mathcal{N}%
(\sigma)^{\alpha^{\prime}/2}\mathcal{N}(\rho)^{-\alpha^{\prime}}%
\mathcal{N}(\sigma)^{\alpha^{\prime}/2}\right)  \sigma^{-\alpha^{\prime}%
/2}\right\} \\
+\operatorname{Tr}\left\{  \rho^{\alpha}\sigma^{-\alpha^{\prime}/2}%
\mathcal{N}^{\dag}\left(  \left[  \frac{d}{d\alpha}\mathcal{N}(\sigma
)^{\alpha^{\prime}/2}\right]  \mathcal{N}(\rho)^{-\alpha^{\prime}}%
\mathcal{N}(\sigma)^{\alpha^{\prime}/2}\right)  \sigma^{-\alpha^{\prime}%
/2}\right\} \\
+\operatorname{Tr}\left\{  \rho^{\alpha}\sigma^{-\alpha^{\prime}/2}%
\mathcal{N}^{\dag}\left(  \mathcal{N}(\sigma)^{\alpha^{\prime}/2}\left[
\frac{d}{d\alpha}\mathcal{N}(\rho)^{-\alpha^{\prime}}\right]  \mathcal{N}%
(\sigma)^{\alpha^{\prime}/2}\right)  \sigma^{-\alpha^{\prime}/2}\right\} \\
+\operatorname{Tr}\left\{  \rho^{\alpha}\sigma^{-\alpha^{\prime}/2}%
\mathcal{N}^{\dag}\left(  \mathcal{N}(\sigma)^{\alpha^{\prime}/2}%
\mathcal{N}(\rho)^{-\alpha^{\prime}}\left[  \frac{d}{d\alpha}\mathcal{N}%
(\sigma)^{\alpha^{\prime}/2}\right]  \right)  \sigma^{-\alpha^{\prime}%
/2}\right\} \\
+\operatorname{Tr}\left\{  \rho^{\alpha}\sigma^{-\alpha^{\prime}/2}%
\mathcal{N}^{\dag}\left(  \mathcal{N}(\sigma)^{\alpha^{\prime}/2}%
\mathcal{N}(\rho)^{-\alpha^{\prime}}\mathcal{N}(\sigma)^{\alpha^{\prime}%
/2}\right)  \left[  \frac{d}{d\alpha}\sigma^{-\alpha^{\prime}/2}\right]
\right\}
\end{multline}%
\begin{multline}
=\Bigg[\operatorname{Tr}\left\{  \rho^{\alpha}\left[  \log\rho\right]
\sigma^{-\alpha^{\prime}/2}\mathcal{N}^{\dag}\left(  \mathcal{N}%
(\sigma)^{\alpha^{\prime}/2}\mathcal{N}(\rho)^{-\alpha^{\prime}}%
\mathcal{N}(\sigma)^{\alpha^{\prime}/2}\right)  \sigma^{-\alpha^{\prime}%
/2}\right\} \\
-\frac{1}{2}\operatorname{Tr}\left\{  \rho\left[  \log\sigma\right]
\sigma^{-\alpha^{\prime}/2}\mathcal{N}^{\dag}\left(  \mathcal{N}%
(\sigma)^{\alpha^{\prime}/2}\mathcal{N}(\rho)^{-\alpha^{\prime}}%
\mathcal{N}(\sigma)^{\alpha^{\prime}/2}\right)  \sigma^{-\alpha^{\prime}%
/2}\right\} \\
+\frac{1}{2}\operatorname{Tr}\left\{  \rho\sigma^{-\alpha^{\prime}%
/2}\mathcal{N}^{\dag}\left(  \left[  \log\mathcal{N}(\sigma)\right]
\mathcal{N}(\sigma)^{\alpha^{\prime}/2}\mathcal{N}(\rho)^{-\alpha^{\prime}%
}\mathcal{N}(\sigma)^{\alpha^{\prime}/2}\right)  \sigma^{-\alpha^{\prime}%
/2}\right\} \\
-\operatorname{Tr}\left\{  \rho\sigma^{-\alpha^{\prime}/2}\mathcal{N}^{\dag
}\left(  \mathcal{N}(\sigma)^{\alpha^{\prime}/2}\left[  \log\mathcal{N}%
(\rho)\right]  \mathcal{N}(\rho)^{-\alpha^{\prime}}\mathcal{N}(\sigma
)^{\alpha^{\prime}/2}\right)  \sigma^{-\alpha^{\prime}/2}\right\} \\
+\frac{1}{2}\operatorname{Tr}\left\{  \rho\sigma^{-\alpha^{\prime}%
/2}\mathcal{N}^{\dag}\left(  \mathcal{N}(\sigma)^{\alpha^{\prime}%
/2}\mathcal{N}(\rho)^{-\alpha^{\prime}}\mathcal{N}(\sigma)^{\alpha^{\prime}%
/2}\left[  \log\mathcal{N}(\sigma)\right]  \right)  \sigma^{-\alpha^{\prime
}/2}\right\} \\
-\frac{1}{2}\operatorname{Tr}\left\{  \rho\sigma^{-\alpha^{\prime}%
/2}\mathcal{N}^{\dag}\left(  \mathcal{N}(\sigma)^{\alpha^{\prime}%
/2}\mathcal{N}(\rho)^{-\alpha^{\prime}}\mathcal{N}(\sigma)^{\alpha^{\prime}%
/2}\right)  \sigma^{-\alpha^{\prime}/2}\left[  \log\sigma\right]  \right\}
\Bigg].
\end{multline}
Taking the limit as $\alpha\rightarrow1$ gives%
\begin{multline}
\left.  \frac{d}{d\alpha}Q_{\alpha}(\rho,\sigma,\mathcal{N})\right\vert
_{\alpha=1}=\operatorname{Tr}\left\{  \rho\left[  \log\rho\right]  \Pi
_{\sigma}\mathcal{N}^{\dag}\left(  \Pi_{\mathcal{N}(\sigma)}\Pi_{\mathcal{N}%
(\rho)}\Pi_{\mathcal{N}(\sigma)}\right)  \Pi_{\sigma}\right\} \\
-\frac{1}{2}\operatorname{Tr}\left\{  \rho\left[  \log\sigma\right]
\Pi_{\sigma}\mathcal{N}^{\dag}\left(  \Pi_{\mathcal{N}(\sigma)}\Pi
_{\mathcal{N}(\rho)}\Pi_{\mathcal{N}(\sigma)}\right)  \Pi_{\sigma}\right\} \\
+\frac{1}{2}\operatorname{Tr}\left\{  \rho\Pi_{\sigma}\mathcal{N}^{\dag
}\left(  \left[  \log\mathcal{N}(\sigma)\right]  \Pi_{\mathcal{N}(\sigma)}%
\Pi_{\mathcal{N}(\rho)}\Pi_{\mathcal{N}(\sigma)}\right)  \Pi_{\sigma}\right\}
\\
-\operatorname{Tr}\left\{  \rho\Pi_{\sigma}\mathcal{N}^{\dag}\left(
\Pi_{\mathcal{N}(\sigma)}\left[  \log\mathcal{N}(\rho)\right]  \Pi
_{\mathcal{N}(\rho)}\Pi_{\mathcal{N}(\sigma)}\right)  \Pi_{\sigma}\right\} \\
+\frac{1}{2}\operatorname{Tr}\left\{  \rho\Pi_{\sigma}\mathcal{N}^{\dag
}\left(  \Pi_{\mathcal{N}(\sigma)}\Pi_{\mathcal{N}(\rho)}\Pi_{\mathcal{N}%
(\sigma)}\left[  \log\mathcal{N}(\sigma)\right]  \right)  \Pi_{\sigma}\right\}
\\
-\frac{1}{2}\operatorname{Tr}\left\{  \rho\Pi_{\sigma}\mathcal{N}^{\dag
}\left(  \Pi_{\mathcal{N}(\sigma)}\Pi_{\mathcal{N}(\rho)}\Pi_{\mathcal{N}%
(\sigma)}\right)  \left[  \log\sigma\right]  \Pi_{\sigma}\right\}  .
\end{multline}
We now simplify the first four terms and note that the last two are Hermitian
conjugates of the second and third:%
\begin{multline}
\operatorname{Tr}\left\{  \rho\left[  \log\rho\right]  \Pi_{\sigma}%
\mathcal{N}^{\dag}\left(  \Pi_{\mathcal{N}(\sigma)}\Pi_{\mathcal{N}(\rho)}%
\Pi_{\mathcal{N}(\sigma)}\right)  \Pi_{\sigma}\right\}  =\operatorname{Tr}%
\left\{  \rho\left[  \log\rho\right]  \mathcal{N}^{\dag}\left(  \Pi
_{\mathcal{N}(\rho)}\right)  \right\} \\
=\operatorname{Tr}\left\{  \mathcal{N}\left(  \rho\left[  \log\rho\right]
\right)  \Pi_{\mathcal{N}(\rho)}\right\}  =\operatorname{Tr}\left\{
U\rho\left[  \log\rho\right]  U^{\dag}\left(  \Pi_{\mathcal{N}(\rho)}\otimes
I_{E}\right)  \right\} \\
=\operatorname{Tr}\left\{  \Pi_{U\rho U^{\dag}}U\rho\left[  \log\rho\right]
U^{\dag}\left(  \Pi_{\mathcal{N}(\rho)}\otimes I_{E}\right)  \right\}
=\operatorname{Tr}\left\{  \Pi_{U\rho U^{\dag}}U\rho\left[  \log\rho\right]
U^{\dag}\right\} \\
=\operatorname{Tr}\left\{  \rho\left[  \log\rho\right]  \right\}  ,
\end{multline}%
\begin{multline}
\operatorname{Tr}\left\{  \rho\left[  \log\sigma\right]  \Pi_{\sigma
}\mathcal{N}^{\dag}\left(  \Pi_{\mathcal{N}(\sigma)}\Pi_{\mathcal{N}(\rho)}%
\Pi_{\mathcal{N}(\sigma)}\right)  \Pi_{\sigma}\right\}  =\operatorname{Tr}%
\left\{  \rho\left[  \log\sigma\right]  \mathcal{N}^{\dag}\left(
\Pi_{\mathcal{N}(\rho)}\right)  \right\} \\
=\operatorname{Tr}\left\{  \mathcal{N}\left(  \rho\left[  \log\sigma\right]
\right)  \left(  \Pi_{\mathcal{N}(\rho)}\right)  \right\}  =\operatorname{Tr}%
\left\{  U\rho\left[  \log\sigma\right]  U^{\dag}\left(  \Pi_{\mathcal{N}%
(\rho)}\otimes I_{E}\right)  \right\} \\
=\operatorname{Tr}\left\{  \Pi_{U\rho U^{\dag}}U\rho U^{\dag}U\left[
\log\sigma\right]  U^{\dag}\left(  \Pi_{\mathcal{N}(\rho)}\otimes
I_{E}\right)  \right\}  =\operatorname{Tr}\left\{  U\rho U^{\dag}U\left[
\log\sigma\right]  U^{\dag}\right\} \\
=\operatorname{Tr}\left\{  \rho\left[  \log\sigma\right]  \right\}  ,
\end{multline}%
\begin{multline}
\operatorname{Tr}\left\{  \rho\Pi_{\sigma}\mathcal{N}^{\dag}\left(  \left[
\log\mathcal{N}(\sigma)\right]  \Pi_{\mathcal{N}(\sigma)}\Pi_{\mathcal{N}%
(\rho)}\Pi_{\mathcal{N}(\sigma)}\right)  \Pi_{\sigma}\right\}
=\operatorname{Tr}\left\{  \rho\mathcal{N}^{\dag}\left(  \left[
\log\mathcal{N}(\sigma)\right]  \Pi_{\mathcal{N}(\rho)}\right)  \right\} \\
=\operatorname{Tr}\left\{  \mathcal{N}(\rho)\left[  \log\mathcal{N}%
(\sigma)\right]  \Pi_{\mathcal{N}(\rho)}\right\}  =\operatorname{Tr}\left\{
\mathcal{N}(\rho)\left[  \log\mathcal{N}(\sigma)\right]  \right\}  ,
\end{multline}%
\begin{multline}
\operatorname{Tr}\left\{  \rho\Pi_{\sigma}\mathcal{N}^{\dag}\left(
\Pi_{\mathcal{N}(\sigma)}\left[  \log\mathcal{N}(\rho)\right]  \Pi
_{\mathcal{N}(\rho)}\Pi_{\mathcal{N}(\sigma)}\right)  \Pi_{\sigma}\right\}
=\operatorname{Tr}\left\{  \rho\mathcal{N}^{\dag}\left(  \left[
\log\mathcal{N}(\rho)\right]  \Pi_{\mathcal{N}(\rho)}\right)  \right\} \\
=\operatorname{Tr}\left\{  \mathcal{N}(\rho)\left(  \left[  \log
\mathcal{N}(\rho)\right]  \Pi_{\mathcal{N}(\rho)}\right)  \right\}
=\operatorname{Tr}\left\{  \mathcal{N}(\rho)\left[  \log\mathcal{N}%
(\rho)\right]  \right\}  .
\end{multline}
This then implies that the following equality holds%
\begin{multline}
\left.  \frac{d}{d\alpha}Q_{\alpha}(\rho,\sigma,\mathcal{N})\right\vert
_{\alpha=1}=\operatorname{Tr}\left\{  \rho\left[  \log\rho\right]  \right\}
-\operatorname{Tr}\left\{  \rho\left[  \log\sigma\right]  \right\}
\label{eq:last-deriv}\\
+\operatorname{Tr}\left\{  \mathcal{N}(\rho)\left[  \log\mathcal{N}%
(\sigma)\right]  \right\}  -\operatorname{Tr}\left\{  \mathcal{N}(\rho)\left[
\log\mathcal{N}(\rho)\right]  \right\}  .
\end{multline}
Putting together (\ref{eq:limit-a-1-exp}) and (\ref{eq:last-deriv}), we can
then conclude the statement of the theorem.
\end{proof}

\section{Auxiliary lemmas and proofs}

\label{app:auxiliary}

\begin{lemma}
\label{lem:max-continuous} Let $\mathcal{A}$ and $\mathcal{T}$ be compact
metric spaces, and let $f:\mathcal{A}\times\mathcal{T}\rightarrow\mathbb{R}$
be a continuous function. Then, $g,h:\mathcal{A}\rightarrow\mathbb{R}$,
defined as $g(\alpha)=\max_{t\in\mathcal{T}}f(\alpha,t)$ and $h(\alpha
)=\min_{t\in\mathcal{T}}f(\alpha,t)$ are continuous.
\end{lemma}

\begin{proof}
By the Heine-Cantor theorem, $f$ is uniformly continuous. Hence, for every
$\varepsilon>0$, there exists a $\delta>0$ such that $\left\vert
f(\alpha,t)-f(\alpha^{\prime},t^{\prime})\right\vert <\varepsilon$ whenever
$D_{\mathcal{A}}(\alpha,\alpha^{\prime})<\delta$ and $D_{\mathcal{T}%
}(t,t^{\prime})<\delta$, where $D_{\mathcal{A}}$ and $D_{\mathcal{T}}$ are the
distance functions on $\mathcal{A}$ and $\mathcal{T}$ respectively. Now, given
$\alpha\in\mathcal{A}$, let $t$ be such that $g(\alpha)=f(\alpha,t)$. Then,
for any $\alpha^{\prime}\in\mathcal{A}$ with $D_{\mathcal{A}}(\alpha
,\alpha^{\prime})<\delta$ we have that
\[
g(\alpha)=f(\alpha,t)<f(\alpha^{\prime},t)+\varepsilon\leqslant\max
_{t^{\prime}\in\mathcal{T}}f(\alpha^{\prime},t^{\prime})+\varepsilon
=g(\alpha^{\prime})+\varepsilon.
\]
By symmetry, we then have that $\left\vert g(\alpha)-g(\alpha^{\prime
})\right\vert <\varepsilon$, which proves the continuity of $g$. A similar
argument establishes the continuity of $h$.
\end{proof}

\bigskip

\begin{proof}
[Proof of Theorem~\ref{thm:rel-ent-other}]Let $\rho$, $\sigma$, and
$\mathcal{N}$ be as given in Definition~\ref{def:rho-sig-N-2}. Let%
\begin{equation}
G\left(  z\right)  =\left(  \mathcal{N}(\rho)^{-z/2}\mathcal{N}(\sigma
)^{z/2}\otimes I_{E}\right)  U\sigma^{-z/2}\rho^{\left(  1+z\right)  /2}.
\end{equation}
In the equation%
\begin{equation}
\frac{1}{p_{\theta}}=\frac{\theta}{p_{0}}+\frac{1-\theta}{p_{1}},
\end{equation}
choose $p_{0}=2$ and $p_{1}=2$, so that $p_{\theta}=2$. Recalling that%
\begin{equation}
M_{k}=\sup_{t\in\mathbb{R}}\left\Vert G\left(  k+it\right)  \right\Vert
_{p_{k}},
\end{equation}
for $k=0,1$, we find that%
\begin{equation}
\left\Vert G\left(  \theta\right)  \right\Vert _{p_{\theta}}\leq
M_{0}^{1-\theta}M_{1}^{\theta}.
\end{equation}
For our choices, we find that%
\begin{align}
M_{0} &  =\sup_{t\in\mathbb{R}}\left\Vert G\left(  it\right)  \right\Vert
_{2}\\
&  =\sup_{t\in\mathbb{R}}\left\Vert \left(  \mathcal{N}(\rho)^{-it/2}%
\mathcal{N}(\sigma)^{it/2}\otimes I_{E}\right)  U\sigma^{-it/2}\rho^{\left(
1+it\right)  /2}\right\Vert _{2}\\
&  =\left\Vert \rho^{1/2}\right\Vert _{2}=1,\\
M_{1} &  =\sup_{t\in\mathbb{R}}\left\Vert G\left(  1+it\right)  \right\Vert
_{2}\\
&  =\sup_{t\in\mathbb{R}}\left\Vert \left(  \mathcal{N}(\rho)^{-\left(
1+it\right)  /2}\mathcal{N}(\sigma)^{\left(  1+it\right)  /2}\otimes
I_{E}\right)  U\sigma^{-\left(  1+it\right)  /2}\rho^{\left(  1+\left(
1+it\right)  \right)  /2}\right\Vert _{2}\\
&  =\sup_{t\in\mathbb{R}}\left\Vert \left(  \mathcal{N}(\rho)^{-1/2}%
\mathcal{N}(\sigma)^{it/2}\mathcal{N}(\sigma)^{1/2}\otimes I_{E}\right)
U\sigma^{-1/2}\sigma^{-it/2}\rho\right\Vert _{2}\\
&  =\left[  \exp\sup_{t\in\mathbb{R}}D_{2}\left(  \rho\Vert\left(
\mathcal{U}_{\sigma,-t}\circ\mathcal{P}_{\sigma,\mathcal{N}}\circ
\mathcal{U}_{\mathcal{N}(\sigma),t}\right)  \left(  \mathcal{N}(\rho)\right)
\right)  \right]  ^{1/2}.
\end{align}
Applying the three-line theorem gives%
\begin{multline}
\left\Vert \left(  \mathcal{N}(\rho)^{-\theta/2}\mathcal{N}(\sigma)^{\theta
/2}\otimes I_{E}\right)  U\sigma^{-\theta/2}\rho^{\left(  1+\theta\right)
/2}\right\Vert _{2}\\
\leq\left[  \exp\sup_{t\in\mathbb{R}}D_{2}\left(  \rho\Vert\left(
\mathcal{U}_{\sigma,-t}\circ\mathcal{P}_{\sigma,\mathcal{N}}\circ
\mathcal{U}_{\mathcal{N}(\sigma),t}\right)  \left(  \mathcal{N}(\rho)\right)
\right)  \right]  ^{\theta/2},
\end{multline}
and after a logarithm gives%
\begin{equation}
\frac{2}{\theta}\log\left\Vert \left(  \mathcal{N}(\rho)^{-\theta
/2}\mathcal{N}(\sigma)^{\theta/2}\otimes I_{E}\right)  U\sigma^{-\theta/2}%
\rho^{\left(  1+\theta\right)  /2}\right\Vert _{2}\leq\sup_{t\in\mathbb{R}%
}D_{2}\left(  \rho\Vert\left(  \mathcal{U}_{\sigma,-t}\circ\mathcal{P}%
_{\sigma,\mathcal{N}}\circ\mathcal{U}_{\mathcal{N}(\sigma),t}\right)  \left(
\mathcal{N}(\rho)\right)  \right)  .
\end{equation}
Take the limit as $\theta\searrow0$ to get%
\begin{equation}
D(\rho\Vert\sigma)-D\left(  \mathcal{N(}\rho)\Vert\mathcal{N(}\sigma)\right)
\leq\sup_{t\in\mathbb{R}}D_{2}\left(  \rho\Vert\left(  \mathcal{U}_{\sigma
,-t}\circ\mathcal{P}_{\sigma,\mathcal{N}}\circ\mathcal{U}_{\mathcal{N}%
(\sigma),t}\right)  \left(  \mathcal{N}(\rho)\right)  \right)  .
\end{equation}

Now we prove the other inequality. Let $\rho$, $\sigma$, and $\mathcal{N}$ be
as given in Definition~\ref{def:rho-sig-N} and such that $\operatorname{supp}%
(\rho)\subseteq\operatorname{supp}(\sigma)$. Take%
\begin{equation}
G\left(  z\right)  =\left(  \mathcal{N}(\rho)^{z/2}\mathcal{N}(\sigma
)^{-z/2}\otimes I_{E}\right)  U\sigma^{z/2}\rho^{\left(  1-z\right)  /2}.
\end{equation}
Then $M_{0}=1$ again and%
\begin{align}
M_{1} &  =\sup_{t\in\mathbb{R}}\left\Vert G\left(  1+it\right)  \right\Vert
_{2}\\
&  =\sup_{t\in\mathbb{R}}\left\Vert \left(  \mathcal{N}(\rho)^{\left(
1+it\right)  /2}\mathcal{N}(\sigma)^{-\left(  1+it\right)  /2}\otimes
I_{E}\right)  U\sigma^{\left(  1+it\right)  /2}\rho^{\left(  1-\left(
1+it\right)  \right)  /2}\right\Vert _{2}\\
&  =\sup_{t\in\mathbb{R}}\left\Vert \left(  \mathcal{N}(\rho)^{1/2}%
\mathcal{N}(\sigma)^{-it/2}\mathcal{N}(\sigma)^{-1/2}\otimes I_{E}\right)
U\sigma^{1/2}\sigma^{it/2}\rho^{0}\right\Vert _{2}\\
&  =\exp\left\{  -\inf_{t\in\mathbb{R}}D_{0}\left(  \rho\Vert\left(
\mathcal{U}_{\sigma,-t}\circ\mathcal{P}_{\sigma,\mathcal{N}}\circ
\mathcal{U}_{\mathcal{N}(\sigma),t}\right)  \left(  \mathcal{N}(\rho)\right)
\right)  \right\}  ^{1/2}.
\end{align}
Applying the three-line theorem gives%
\begin{multline}
\left\Vert \left(  \mathcal{N}(\rho)^{\theta/2}\mathcal{N}(\sigma)^{-\theta
/2}\otimes I_{E}\right)  U\sigma^{\theta/2}\rho^{\left(  1-\theta\right)
/2}\right\Vert _{2}\\
\leq\left[  \exp\left\{  -\inf_{t\in\mathbb{R}}D_{0}\left(  \rho\Vert\left(
\mathcal{U}_{\sigma,-t}\circ\mathcal{P}_{\sigma,\mathcal{N}}\circ
\mathcal{U}_{\mathcal{N}(\sigma),t}\right)  \left(  \mathcal{N}(\rho)\right)
\right)  \right\}  \right]  ^{\theta/2},
\end{multline}
which after taking a logarithm gives%
\begin{equation}
\frac{2}{-\theta}\log\left\Vert \left(  \mathcal{N}(\rho)^{\theta
/2}\mathcal{N}(\sigma)^{-\theta/2}\otimes I_{E}\right)  U\sigma^{\theta/2}%
\rho^{\left(  1-\theta\right)  /2}\right\Vert _{2}\geq\inf_{t\in\mathbb{R}%
}D_{0}\left(  \rho\Vert\left(  \mathcal{U}_{\sigma,-t}\circ\mathcal{P}%
_{\sigma,\mathcal{N}}\circ\mathcal{U}_{\mathcal{N}(\sigma),t}\right)  \left(
\mathcal{N}(\rho)\right)  \right)  .
\end{equation}
Take the limit as $\theta\searrow0$ to get%
\begin{equation}
D(\rho\Vert\sigma)-D\left(  \mathcal{N(}\rho)\Vert\mathcal{N(}\sigma)\right)
\geq\inf_{t\in\mathbb{R}}D_{0}\left(  \rho\Vert\left(  \mathcal{U}_{\sigma
,-t}\circ\mathcal{P}_{\sigma,\mathcal{N}}\circ\mathcal{U}_{\mathcal{N}%
(\sigma),t}\right)  \left(  \mathcal{N}(\rho)\right)  \right)  .
\end{equation}

\end{proof}

\section{Taylor expansions}

\label{app:taylor}Here we show the following limit:%
\begin{equation}
\lim_{\alpha\rightarrow1}f\left(  \alpha,V_{\mathcal{N}(\sigma)},V_{\sigma
}\right)  =f\left(  1,V_{\mathcal{N}(\sigma)},V_{\sigma}\right)  ,
\end{equation}
where $f\left(  \alpha,V_{\mathcal{N}(\sigma)},V_{\sigma}\right)  $ is defined
as%
\begin{equation}
f\left(  \alpha,V_{\mathcal{N}(\sigma)},V_{\sigma}\right)  =\frac{1}{\alpha
-1}\log\left\Vert \left(  \left[  \mathcal{N}\left(  \rho\right)  \right]
^{\left(  1-\alpha\right)  /2}V_{\mathcal{N}(\sigma)}\left[  \mathcal{N}%
(\sigma)\right]  ^{\left(  \alpha-1\right)  /2}\otimes I_{E}\right)
U\sigma^{\left(  1-\alpha\right)  /2}V_{\sigma}\rho^{\alpha/2}\right\Vert
_{2}^{2}%
\end{equation}
and $f\left(  1,V_{\mathcal{N}(\sigma)},V_{\sigma}\right)  $ in
(\ref{eq:f-function}). From the fact that%
\begin{equation}
\left.  \log\left\Vert \left(  \left[  \mathcal{N}(\rho)\right]  ^{\left(
1-\alpha\right)  /2}V_{\mathcal{N}(\sigma)}\left[  \mathcal{N}(\sigma)\right]
^{\left(  \alpha-1\right)  /2}\otimes I_{E}\right)  U\sigma^{\left(
1-\alpha\right)  /2}V_{\sigma}\rho^{\alpha/2}\right\Vert _{2}^{2}\right\vert
_{\alpha=1}=0,
\end{equation}
we know (from the definition of derivative) that $\lim_{\alpha\rightarrow
1}f\left(  \alpha,V_{\mathcal{N}(\sigma)},V_{\sigma}\right)  $ is equal to%
\begin{multline}
\left.  \frac{d}{d\alpha}\log\left\Vert \left(  \left[  \mathcal{N}\left(
\rho\right)  \right]  ^{\left(  1-\alpha\right)  /2}V_{\mathcal{N}(\sigma
)}\left[  \mathcal{N}(\sigma)\right]  ^{\left(  \alpha-1\right)  /2}\otimes
I_{E}\right)  U\sigma^{\left(  1-\alpha\right)  /2}V_{\sigma}\rho^{\alpha
/2}\right\Vert _{2}^{2}\right\vert _{\alpha=1}\\
=\left.  \frac{d}{d\alpha}\left\Vert \left(  \left[  \mathcal{N}\left(
\rho\right)  \right]  ^{\left(  1-\alpha\right)  /2}V_{\mathcal{N}(\sigma
)}\left[  \mathcal{N}(\sigma)\right]  ^{\left(  \alpha-1\right)  /2}\otimes
I_{E}\right)  U\sigma^{\left(  1-\alpha\right)  /2}V_{\sigma}\rho^{\alpha
/2}\right\Vert _{2}^{2}\right\vert _{\alpha=1}.
\end{multline}
We evaluate the latter derivative by employing Taylor expansions. Substitute
$\alpha=1+\gamma$, so that the quantity inside the derivative operation is
equal to%
\begin{equation}
\left\Vert \left(  \left[  \mathcal{N}(\rho)\right]  ^{-\gamma/2}%
V_{\mathcal{N}(\sigma)}\left[  \mathcal{N}(\sigma)\right]  ^{\gamma/2}\otimes
I_{E}\right)  U\sigma^{-\gamma/2}V_{\sigma}\rho^{\left(  1+\gamma\right)
/2}\right\Vert _{2}^{2},
\end{equation}
which we can rewrite as%
\begin{equation}
\left\Vert \left(  \left[  V_{\mathcal{N}(\sigma)}^{\dag}\mathcal{N}%
(\rho)V_{\mathcal{N}(\sigma)}\right]  ^{-\gamma/2}\left[  \mathcal{N}%
(\sigma)\right]  ^{\gamma/2}\otimes I_{E}\right)  U\sigma^{-\gamma/2}\left[
V_{\sigma}\rho V_{\sigma}^{\dag}\right]  ^{\left(  1+\gamma\right)
/2}\right\Vert _{2}^{2},
\end{equation}
due to the unitary invariance of the norm. Now we use that%
\begin{align}
\left[  V_{\sigma}\rho V_{\sigma}^{\dag}\right]  ^{\left(  1+\gamma\right)
/2}  &  =\left[  V_{\sigma}\rho V_{\sigma}^{\dag}\right]  ^{1/2}+\frac{\gamma
}{2}\left[  V_{\sigma}\rho V_{\sigma}^{\dag}\right]  ^{1/2}\log\left[
V_{\sigma}\rho V_{\sigma}^{\dag}\right]  +O\left(  \gamma^{2}\right)  ,\\
\sigma^{-\gamma/2}  &  =I-\frac{\gamma}{2}\log\sigma+O\left(  \gamma
^{2}\right)  ,\\
\left[  \mathcal{N}(\sigma)\right]  ^{\gamma/2}  &  =I+\frac{\gamma}{2}%
\log\left[  \mathcal{N}(\sigma)\right]  +O\left(  \gamma^{2}\right)  ,\\
\left[  V_{\mathcal{N}(\sigma)}^{\dag}\mathcal{N}\left(  \rho\right)
V_{\mathcal{N}(\sigma)}\right]  ^{-\gamma/2}  &  =I-\frac{\gamma}{2}%
\log\left[  V_{\mathcal{N}(\sigma)}^{\dag}\mathcal{N}(\rho)V_{\mathcal{N}%
(\sigma)}\right]  +O\left(  \gamma^{2}\right)  .
\end{align}
The above implies that%
\begin{multline}
\left[  V_{\mathcal{N}(\sigma)}^{\dag}\mathcal{N}\left(  \rho\right)
V_{\mathcal{N}(\sigma)}\right]  ^{-\gamma/2}\left[  \mathcal{N}(\sigma
)\right]  ^{\gamma/2}U\sigma^{-\gamma/2}\left[  V_{\sigma}\rho V_{\sigma
}^{\dag}\right]  ^{\left(  1+\gamma\right)  /2}\\
=\left(  I-\frac{\gamma}{2}\log\left[  V_{\mathcal{N}(\sigma)}^{\dag
}\mathcal{N}(\rho)V_{\mathcal{N}(\sigma)}\right]  \right)  \left(
I+\frac{\gamma}{2}\log\left[  \mathcal{N}(\sigma)\right]  \right)  \times\\
U\left(  I-\frac{\gamma}{2}\log\sigma\right)  \left(  \left[  V_{\sigma}\rho
V_{\sigma}^{\dag}\right]  ^{1/2}+\frac{\gamma}{2}\left[  V_{\sigma}\rho
V_{\sigma}^{\dag}\right]  ^{1/2}\log\left[  V_{\sigma}\rho V_{\sigma}^{\dag
}\right]  \right)  +O\left(  \gamma^{2}\right)  .
\end{multline}
By working out the right-hand side above and neglecting terms of second order
in $\gamma$ and higher, we find that%
\begin{multline}
\left[  V_{\mathcal{N}(\sigma)}^{\dag}\mathcal{N}\left(  \rho\right)
V_{\mathcal{N}(\sigma)}\right]  ^{-\gamma/2}\left[  \mathcal{N}(\sigma
)\right]  ^{\gamma/2}U\sigma^{-\gamma/2}\left[  V_{\sigma}\rho V_{\sigma
}^{\dag}\right]  ^{\left(  1+\gamma\right)  /2}\label{eq:t-expand-1}\\
=U\left[  V_{\sigma}\rho V_{\sigma}^{\dag}\right]  ^{1/2}-\frac{\gamma}{2}%
\log\left[  V_{\mathcal{N}(\sigma)}^{\dag}\mathcal{N}\left(  \rho\right)
V_{\mathcal{N}(\sigma)}\right]  U\left[  V_{\sigma}\rho V_{\sigma}^{\dag
}\right]  ^{1/2}+\frac{\gamma}{2}\log\left[  \mathcal{N}(\sigma)\right]
U\left[  V_{\sigma}\rho V_{\sigma}^{\dag}\right]  ^{1/2}\\
-\frac{\gamma}{2}U\left[  \log\sigma\right]  \left[  V_{\sigma}\rho V_{\sigma
}^{\dag}\right]  ^{1/2}+\frac{\gamma}{2}U\left[  V_{\sigma}\rho V_{\sigma
}^{\dag}\right]  ^{1/2}\log\left[  V_{\sigma}\rho V_{\sigma}^{\dag}\right]
+O\left(  \gamma^{2}\right)  .
\end{multline}
The Hermitian conjugate is%
\begin{multline}
\left[  V_{\sigma}\rho V_{\sigma}^{\dag}\right]  ^{1/2}U^{\dag}-\frac{\gamma
}{2}\left[  V_{\sigma}\rho V_{\sigma}^{\dag}\right]  ^{1/2}U^{\dag}\log\left[
V_{\mathcal{N}(\sigma)}^{\dag}\mathcal{N}(\rho)V_{\mathcal{N}(\sigma)}\right]
+\frac{\gamma}{2}\left[  V_{\sigma}\rho V_{\sigma}^{\dag}\right]
^{1/2}U^{\dag}\log\left[  \mathcal{N}(\sigma)\right] \\
-\frac{\gamma}{2}\left[  V_{\sigma}\rho V_{\sigma}^{\dag}\right]
^{1/2}\left[  \log\sigma\right]  U^{\dag}+\frac{\gamma}{2}\left[  \log\left[
V_{\sigma}\rho V_{\sigma}^{\dag}\right]  \right]  \left[  V_{\sigma}\rho
V_{\sigma}^{\dag}\right]  ^{1/2}U^{\dag}+O\left(  \gamma^{2}\right)  .
\end{multline}
Combining (\ref{eq:t-expand-1}) with its Hermitian conjugate and neglecting
higher order terms gives%
\begin{multline}
\left[  V_{\sigma}\rho V_{\sigma}^{\dag}\right]  -\gamma\left[  V_{\sigma}\rho
V_{\sigma}^{\dag}\right]  ^{1/2}\mathcal{N}^{\dag}\left(  \log\left[
V_{\mathcal{N}(\sigma)}^{\dag}\mathcal{N}(\rho)V_{\mathcal{N}(\sigma)}\right]
\right)  \left[  V_{\sigma}\rho V_{\sigma}^{\dag}\right]  ^{1/2}\\
+\gamma\left[  V_{\sigma}\rho V_{\sigma}^{\dag}\right]  ^{1/2}\mathcal{N}%
^{\dag}\left(  \log\left[  \mathcal{N}(\sigma)\right]  \right)  \left[
V_{\sigma}\rho V_{\sigma}^{\dag}\right]  ^{1/2}-\gamma\left[  V_{\sigma}\rho
V_{\sigma}^{\dag}\right]  ^{1/2}\left[  \log\sigma\right]  \left[  V_{\sigma
}\rho V_{\sigma}^{\dag}\right]  ^{1/2}\\
+\frac{\gamma}{2}\left[  V_{\sigma}\rho V_{\sigma}^{\dag}\right]  \log\left[
V_{\sigma}\rho V_{\sigma}^{\dag}\right]  +\frac{\gamma}{2}\left(  \log\left[
V_{\sigma}\rho V_{\sigma}^{\dag}\right]  \right)  \left[  V_{\sigma}\rho
V_{\sigma}^{\dag}\right]  +O\left(  \gamma^{2}\right)  .
\end{multline}
Taking a trace gives%
\begin{multline}
\text{Tr}\left\{  \rho\right\}  -\gamma\text{Tr}\left\{  \left[  V_{\sigma
}\rho V_{\sigma}^{\dag}\right]  \mathcal{N}^{\dag}\left(  \log\left[
V_{\mathcal{N}(\sigma)}^{\dag}\mathcal{N}(\rho)V_{\mathcal{N}(\sigma)}\right]
\right)  \right\} \\
+\gamma\text{Tr}\left\{  \left[  V_{\sigma}\rho V_{\sigma}^{\dag}\right]
\mathcal{N}^{\dag}\left(  \log\left[  \mathcal{N}(\sigma)\right]  \right)
\right\}  -\gamma\text{Tr}\left\{  \rho\left[  \log\sigma\right]  \right\}
+\gamma\text{Tr}\left\{  \rho\log\rho\right\}  +O\left(  \gamma^{2}\right)  .
\end{multline}
We can now finally use the above development to conclude that%
\begin{align}
&  \left.  \frac{d}{d\alpha}\left\Vert \left(  \left[  \mathcal{N}\left(
\rho\right)  \right]  ^{\left(  1-\alpha\right)  /2}V_{\mathcal{N}(\sigma
)}\left[  \mathcal{N}(\sigma)\right]  ^{\left(  \alpha-1\right)  /2}\otimes
I_{E}\right)  U\sigma^{\left(  1-\alpha\right)  /2}V_{\sigma}\rho^{\alpha
/2}\right\Vert _{2}^{2}\right\vert _{\alpha=1}\nonumber\\
&  =\left.  \frac{d}{d\gamma}\left\Vert \left(  \left[  V_{\mathcal{N}%
(\sigma)}^{\dag}\mathcal{N}(\rho)V_{\mathcal{N}(\sigma)}\right]  ^{-\gamma
/2}\left[  \mathcal{N}(\sigma)\right]  ^{\gamma/2}\otimes I_{E}\right)
U\sigma^{-\gamma/2}\left[  V_{\sigma}\rho V_{\sigma}^{\dag}\right]  ^{\left(
1+\gamma\right)  /2}\right\Vert _{2}^{2}\right\vert _{\gamma=0}\\
&  =\operatorname{Tr}\left\{  \rho\left[  \log\rho-\log\sigma\right]
\right\}  -\operatorname{Tr}\left\{  \mathcal{N}\left(  \left[  V_{\sigma}\rho
V_{\sigma}^{\dag}\right]  \right)  \left[  \log\left[  V_{\mathcal{N}(\sigma
)}^{\dag}\mathcal{N}(\rho)V_{\mathcal{N}(\sigma)}\right]  -\log\left[
\mathcal{N}(\sigma)\right]  \right]  \right\} \\
&  =f(1,V_{\mathcal{N}(\sigma)},V_{\sigma}).
\end{align}

A similar development with Taylor expansions leads to the conclusion that
(\ref{eq:g-limit-1}) holds. However here one should employ the method outlined
in the proof of \cite[Proposition 11]{WWY13}.

\bibliographystyle{alpha}
\bibliography{Ref}

\end{document}